\documentclass[a4paper,11pt]{article}

\usepackage{graphicx}
\usepackage{amsmath}
\usepackage{amsfonts}
\usepackage{amsthm}
\usepackage{algorithm}
\usepackage[noend]{algpseudocode}
\usepackage{url}
\usepackage{comment}

\usepackage[margin=1in]{geometry}

\newcommand{\eps}{{\ensuremath{\varepsilon}}}
\newcommand{\supp}{\mathit{Supp}}
\newcommand{\degv}{\ensuremath{\mathcal{D}}}
\newcommand{\cwei}[1][]{\textsc{Caraway}\textsubscript{#1}}
\newcommand{\recovered}{\ensuremath{\mathcal{R}}}
\newcommand{\sketch}{\ensuremath{\mathcal{L}}}

\newcommand{\updates}[1]{{#1}}

\usepackage{enumitem}
\usepackage{comment}
\usepackage{xcolor}

\newtheorem{theorem}{Theorem}
\newtheorem{lemma}{Lemma}
\newtheorem{proposition}{Proposition}
\newtheorem{corollary}{Corollary}

\begin{document}

\title{Sublinear-Space Streaming Algorithms for Estimating Graph Parameters on Sparse Graphs}

\author{Xiuge Chen\thanks{School of Computing and Information Systems, The University of Melbourne,~AU.} \\ \url{xiugechen@gmail.com} \and
Rajesh Chitnis\thanks{School of Computer Science, University of Birmingham, UK.}\\ \url{rajeshchitnis@gmail.com} \and
Patrick~Eades\footnotemark[1]\\ \url{patrick.f.eades@gmail.com}
\and
Anthony Wirth\footnotemark[1]\\ \url{awirth@unimelb.edu.au}}

\maketitle

\begin{abstract}

In this paper, we design sub-linear space streaming algorithms for estimating three fundamental parameters -- maximum independent set, minimum dominating set and maximum matching -- on sparse graph classes, i.e., graphs which satisfy $m=O(n)$ where $m,n$ is the number of edges, vertices respectively. Each of the three graph parameters we consider can have size $\Omega(n)$ even on sparse graph classes, and hence for sublinear-space algorithms we are restricted to parameter estimation instead of attempting to find a solution.
We obtain these results:

\begin{itemize}
    \item
    \textbf{Estimating Max Independent Set via the Caro-Wei bound}:
    Caro and Wei each showed $\lambda = \sum_{v} {1}/(d(v) + 1)$ is a lower bound on max independent set size, where vertex~$v$ has degree~$d(v)$.
    If average degree,~$\bar{d}$, is~$\mathcal{O}(1)$, and max degree~$\Delta = \mathcal{O}(\eps^{2} \bar{d}^{-3} n)$, our  algorithms, with at least~$1 - \delta$ success probability:
    \begin{itemize}
        \item In \emph{online streaming}, return an actual independent set of size~$1 \pm \eps$ times~$\lambda$. 
        This improves on Halld{\'o}rsson et al.\ [Algorithmica '16]: we have less working space, i.e.,~$\mathcal{O}(\log \eps^{-1} \cdot \log n \cdot \log \delta^{-1})$, faster updates, i.e.,~$\mathcal{O}(\log \eps^{-1})$, and bounded success probability.

        \item In \emph{insertion-only} streams, approximate~$\lambda$  within factor~$1 \pm \eps$, in one pass, in \sloppy$\mathcal{O}(\bar{d} \eps^{-2} \log n \cdot \log \delta^{-1})$ space.
        This aligns with the result of Cormode et al.\ [ISCO '18], though our method also works for \emph{online streaming}.
        In a vertex-arrival and random-order stream, space reduces to~$\mathcal{O}(\log (\bar{d} \eps^{-1}))$.
        With extra space and post-processing step, we remove the max-degree constraint.

    \end{itemize}

    \item
    \textbf{Sublinear-Space Algorithms on Forests}:
    On a forest, Esfandiari et al.\ [SODA '15, TALG '18] showed space lower bounds
    for $1$-pass randomized algorithms that approximately estimate these graph parameters.
    We narrow the gap between upper and lower bounds:
    \begin{itemize}
        \item Max independent set size within~$3/2 \cdot (1 \pm \eps)$ in one pass and in~$\log^{\mathcal{O}(1)} n$ space, and within~$4/3\cdot (1 \pm \eps)$ in two passes and in~$\tilde{\mathcal{O}}(\sqrt{n})$ space; the lower bound is for approx.~$\leq 4/3$.

        \item Min dominating set size within~$3 \cdot (1 \pm \eps)$ in one pass and in~$\log^{\mathcal{O}(1)} n$ space, and within~$2\cdot (1 \pm \eps)$ in two passes and in~$\tilde{\mathcal{O}}(\sqrt{n})$ space; the lower bound is for approx.~$\leq 3/2$.

        \item Max matching size within~$2 \cdot (1 \pm \eps)$ in one pass and in~$\log^{\mathcal{O}(1)} n$ space, and within~$3/2\cdot (1 \pm \eps)$ in two passes and in~$\tilde{\mathcal{O}}(\sqrt{n})$ space; the lower bound is for approx.~$\leq 3/2$.
    \end{itemize}

\end{itemize}


\end{abstract}

\section{Introduction}

    Maximum independent set, minimum dominating set, and maximum matching are key graph problems.
    Independent set models, for example, optimization and scheduling problems where conflicts should be avoided~\cite{araujo2011maximum,gemsa2016evaluation,hossain2020automated,kieritz2010distributed}, dominating set models guardian selection problems~\cite{milenkovic2011dominating,nacher2012dominating,pino2018dominating,shen2010multi,yu2013connected}, while a matching models similarity~\cite{das2012finding,ganti2013data,wang2011fast} and inclusion dependencies~\cite{bauckmann2012discovering,deng2017data}.
    When the input is presented as a data stream, we show new algorithms for computing the following parameters related to these problems on sparse graphs.

    \textbf{Parameters of interest:}
    Given an undirected graph~$G$, comprising vertices~$V$ and edges~$E$, let~$n$ and~$m$ be the sizes of~$V$ and~$E$, respectively. A subset~$S$ of~$V$ is an \emph{independent set} if and only if the subgraph induced by~$S$ contains no edges.
    Subset~$S$ is a \emph{dominating set} if and only if every vertex in~$V$ is either in~$S$ or adjacent to some vertex in~$S$. A subset~$M$ of~$E$ is a \emph{matching} if and only if no pair of edges in~$M$ share a vertex.
    The three parameters of interest to us are the size of a maximum independent set, aka \emph{independence number},~$\beta$; the size of a minimum dominating set, aka \emph{domination number},~$\gamma$; and the size of a maximum matching, aka \emph{matching number},~$\phi$. It is well known that~$n - \phi \geq \beta \geq n - 2 \phi$.

    Given some~$k$, it is NP-complete~\cite{karp1972reducibility} to decide whether~$\beta \geq k$ and to decide whether~$\gamma \leq k$.
    There are several \emph{approximation algorithms} for these problems that run in polynomial time and return a solution within a guaranteed factor of optimum.
    For instance, a greedy algorithm for maximum matching outputs a~$2$-approximation.
    Similarly, there exist greedy algorithms that~$\mathcal{O}(\Delta)$-approximate the maximum independent set \cite{halldorsson1997greed} and approximate within~$\mathcal{O}(\ln \Delta)$ the minimum dominating set \cite{johnson1974approximation,lovasz1975ratio}, where~$\Delta$ is the maximum degree. 
    More promising results are obtained on sparse graphs, where~$m \in \mathcal{O}(n)$, which we study in this paper. 
    For example, the Caro-Wei bound~\cite{caro1979new,wei1981lower}~$\lambda = \sum_{v \in V(G)} (1 + \text{deg}(v))^{-1}$ is a lower bound on~$\beta$.

    \textbf{Data streams:}
    We focus on estimating these graph parameters in the data stream model.
    The \emph{semi-streaming} model~\cite{feigenbaum2005graph}, with~$\mathcal{O}(n \log n)$ bits of working space, is commonplace for general graphs.
    In this paper we consider \emph{sparse} graphs:
    as~$\mathcal{O}(n)$ bits would be sufficient to store the entire sparse graph,
    we restrict our space allowance to~$o(n)$ bits. We \emph{tune} our algorithms to specific stream formats:
    edge-arrival, vertex-arrival, insertion-only, turnstile, arbitrary, and random.
    A graph stream is \emph{edge-arrival} if edges arrive one by one, sequentially, in arbitrary order.
    We call a graph stream \emph{insertion-only} if there are no deletions.
    In a \emph{turnstile} stream, an edge can be deleted, but only if its most recent operation was an insertion.
    It is \emph{vertex-arrival} if it comprises a sequence of $(\text{vertex},\text{vertex-list})$ pairs, $(u_i,A_i)$, where the list~$A_i$ comprises the subset of the vertices~$\{u_j\}_{j < i}$ that have occurred previously in the stream that are adjacent to~$u_i$.
    Moreover, the order of vertex arrivals can be either \emph{arbitrary} or (uniformly) \emph{random}.
    
    Halld{\'o}rsson et al.~\cite{halldorsson2016streaming} introduced the \emph{online streaming} model, combining the data-stream and online models.
    In \emph{online streaming}, after each stream item, on demand, the algorithm must efficiently report a valid solution. Typically, an online-streaming algorithm has a initial solution and modifies it element by element. Similar to the \emph{online} model, each decision on the solution is irrevocable.
    We distinguish between \emph{working space}, involved in computing the solution, and an additional \emph{solution space} for storing or returning
    the solution. Solution space may be significantly larger than the working space, but is write only.

    \subsection{Previous Results}

        To further set the scene for our contributions, we describe some of the streaming-algorithm context.
        On graphs with bounded arboricity,~$\alpha$, there are known approximation algorithms for estimating the matching number,~$\phi$. For insertion-only streams,~Esfandiari et al.~\cite{esfandiari2018streaming} developed a~$(5 \alpha + 9)$-approximation algorithm that requires~$\tilde{\mathcal{O}}(\alpha n^{2/3})$ space. Cormode et al.~\cite{cormode2017sparse} traded off approximation ratio for space, and showed a~$(22.5 \alpha + 6)$-approximation algorithm in only~$\mathcal{O}(\alpha \log^{\mathcal{O}(1)} n)$ space. McGregor and Vorotnikova~\cite{mcgregor2018simple} showed that the algorithm of Cormode et al.~\cite{cormode2017sparse} can achieve a~$(\alpha + 2)$-approximation via a tighter analysis. For turnstile streams, Chitnis et al.~\cite{chitnis2016kernelization} designed a~$(22.5 \alpha + 6)$-approximation algorithm using~$\tilde{\mathcal{O}}(\alpha n^{4/5})$ space. Bury et al.~\cite{bury2019structural} improved the approximation ratio to~$(\alpha + 2)$, their algorithm requires~$\tilde{\mathcal{O}}(\alpha n^{4/5})$ space in edge-arrival streams, but only~$\mathcal{O}(\log n)$ in vertex-arrival streams.

        Switching the arrival order from arbitrary to random, researchers have developed algorithms for domination number~$\gamma$ and independence number~$\beta$. 
        Monemizadeh et al.~\cite{monemizadeh2017testable} converted known constant-time RAM-model approximation algorithms into constant-space streaming algorithms. Their algorithm approximates the domination number in bounded-degree graphs, with an additive error term of~$\eps n$. 
        Peng and Sohler~\cite{peng2018estimating} further extended this result to graphs of bounded average degree.
        They also showed that in planar graphs, independence number~$\beta$ can be~$(1 + \eps)$-approximated using constant, but still massive, space,  i.e.,~$\mathcal{O}(2^{(1/\eps)^{(1/\eps)^{\log^{\mathcal{O}(1)} (1/\eps)}}})$.

        Another line of research approximates independence number via the Caro-Wei Bound. 
        Halld{\'o}rsson et al.~\cite{halldorsson2016streaming} studied general hypergraphs and gave a one-pass insertion-only streaming algorithm in~$\mathcal{O}(n)$ space, outputting, in expectation, an independent set with size at least the Caro-Wei bound,~$\lambda$; 
        their algorithm suits the \emph{online streaming} model. 
        Cormode et al.~\cite{cormode2018approximating} designed an algorithm that~$(1 \pm \eps)$-approximates~$\lambda$ with constant success probability in~$\mathcal{O}(\eps^{-2} \bar{d} \log n)$ space.
        When the input is a vertex-arrival stream with very large average degree, they achieved an~$(\log n)$-approximation with~$\mathcal{O}(\log^3 n)$ space.
        They also showed a nearly tight lower bound: every randomized one-pass algorithm with constant error probability requires~$\Omega({\eps^{-2}}{\bar{d}})$ space to~$(1 \pm \eps)$-approximate~$\lambda$. 

        Meanwhile, Chitnis and Cormode~\cite{chitnis2019towards} adapted the lower-bound reduction technique from Assadi et al.~\cite{assadi2019tight} and showed that, for graphs with arboricity~$\alpha + 2$, for~$\alpha \geq 1$, every randomized~${\alpha}/{32}$-approximation algorithm for minimum dominating set requires~$\Omega(n)$ space. This lower bound holds even under the vertex-arrival model.
    
        Simplifying the graph does not make things significantly easier. 
        When the input is a tree, several~$(2 + \eps)$-approximation algorithms are known for matching number,~$\phi$. 
        Esfandiari et al.~\cite{esfandiari2018streaming} designed an~$\tilde{\mathcal{O}}(\sqrt{n})$-space algorithm for insertion-only streams, where~$\tilde{\mathcal{O}}$ notation suppresses a poly-logarithmic factor. 
        The space was further reduced to~$\log^{\mathcal{O}(1)} n$ by~Cormode et al.~\cite{cormode2017sparse}, while~Bury et al.~\cite{bury2019structural} generalized it to turnstile streams. 
        
        There is relatively little research on independent set and dominating set in trees or forests. 
        Esfandiari et al.~\cite{esfandiari2018streaming} established space lower bounds for estimating~$\phi$ in a forest:
        every $1$-pass randomized streaming approximation algorithm with factor better than~$3/2$ needs~$\Omega(\sqrt{n})$ space. 
        Adapting their approach gives other lower bounds:~$\Omega(\sqrt{n})$ bits are required to approximate~$\beta$ better than~$4/3$, and~$\gamma$ better than~$3/2$.

    \subsection{Our Results}
    \label{sec:our-results}

        We design new sublinear-space streaming algorithms for estimating~$\beta$,~$\gamma$,~$\phi$ in sparse graphs. With such little space, parameter estimation itself is of considerable importance.
        Our results fall into two categories. 
        First, approximating independent set via the Caro-Wei bound in bounded-average-degree graphs%
        \footnote{This includes planar graphs, bounded treewidth, bounded genus, $H$-minor-free, etc.}. 
        Second, approximating parameters on streamed forests. 

        \paragraph{\textbf{Approximating the Caro-Wei Bound ($\lambda$)}}
        \label{sec:our-results-caro-wei}

            Boppana et al.~\cite{boppana2018simple} show that~$\lambda$ is at least the Tur{\'a}n Bound~\cite{turan1941on}, i.e.,~$\beta \geq \lambda \geq {n}/(\bar{d} + 1)$, indicating every~$\mathcal{O}(1)$-estimate of~$\lambda$ is a~$\mathcal{O}(\bar{d})$-approximation of~$\beta$.
            We hence approximate~$\beta$ in bounded-average-degree graphs; we summarize results in Table~\ref{table:ind-num-average-degree}.
            
            With a random permutation of vertices, there is a known offline algorithm for estimating~$\lambda$.
            Our algorithm, \cwei, simulates such a permutation via hash functions drawn uniformly at random from an $\eps$-min-wise hash family.
            Let~$\eps$ be the approximation error \emph{rate}%
            \footnote{The \emph{relative error} between the estimate and the actual value.}.
            For graphs with average degree~$\bar{d}$ and max degree~$\Delta \in \mathcal{O}(\eps^{2} \bar{d}^{-3} n)$,~\cwei\ $(1 \pm \eps)$-approximates~$\lambda$ in an arbitrary-order edge-arrival stream. It fails with probability at most~$\delta$, and uses working space~$\mathcal{O}(\bar{d} \eps^{-2} \log n \cdot \log \delta^{-1})$. 
            By allowing~$\mathcal{O}(\log \eps^{-1} \cdot \log n \cdot \log \delta^{-1})$ additional, write-only, \emph{solution space} to store the output, the modified algorithm, \cwei[1], 
            reports an actual solution set in the \emph{online streaming} model with $\mathcal{O}(\log \eps^{-1})$ update time. 
            The max-degree constraint is required only to bound the failure probability,~$\delta$ for these two algorithms; if an \emph{in-expectation} bound suffices, this constraint can be disregarded.
            
            Our other modified algorithm, \cwei[2], removes the max-degree constraint entirely, but requires more space. 
            Failing with low probability, \cwei[2] returns a~$(1 \pm \eps)$-approximation using~$\mathcal{O}(\bar{d}^2 \eps^{-2} \log^2 n)$ space. Since post-processing is required, \cwei[2] works in the standard streaming model, but not in the \emph{online streaming} model. Additionally, if the stream is vertex-arrival and random-order,~\cwei\ can be modified to use only~$\mathcal{O}(\log (\bar{d} \eps^{-1})\cdot \log \delta^{-1})$ space.

            Bounding \emph{rate}~$\eps$ is non-trivial: there is positive correlation between vertices being in the sample whence we estimate~$\lambda$. To bound~$\eps$, we start by constraining the max degree;
            we also bound~$\eps$ in the offline algorithm and all  its variants.

            \textbf{Comparison with Cormode et al.~\cite{cormode2018approximating}\ }: 
            Since Cormode et al.~relied on estimating sampled vertices' degrees, we understand their methods report the parameter~$\beta$, but not an actual independent set.
            In an insertion-only stream satisfying our~$\Delta$ constraint, \cwei\ has asymptotically the same approximation ratio in the same space as Cormode et al. Importantly, in \emph{online streaming}, \cwei[1], can output an actual independent set.
            Moreover, in vertex-arrival and random-order streams, \cwei\ requires comparatively less space.

            \textbf{Comparison with Halld{\'o}rsson et al.~\cite{halldorsson2016streaming}\ }: The online streaming algorithm by Halld{\'o}rsson et al.~has~$\mathcal{O}(\log n)$ update time and~$\mathcal{O}(n)$ working space. \cwei[1] has faster update time, \sloppy$\mathcal{O}(\log \eps^{-1})$, and less working space,~$\mathcal{O}(\log \eps^{-1} \cdot \log n \cdot \log \delta^{-1})$.
            Besides, our estimate is within an error rate,~$\eps$, with a guaranteed constant probability; there is no such guarantee from Halld{\'o}rsson et al.

            \paragraph*{\textbf{Estimating~$\beta, \gamma$ and $\phi$ on Forests}}
            Table~\ref{table:alg-tree-forest} has results on streamed forests. 
            Esfandiari et al.~\cite{esfandiari2018streaming} showed it is non-trivial to estimate fundamental graph parameters (e.g., $\beta$,~$\gamma$,~$\phi$) in a one-pass streamed forest \footnote{$\Omega(\sqrt{n})$ (or $\Omega(n)$) space is required for randomized (or deterministic) algorithms.}).
            We trade off approximation ratio for space and number of passes, and obtain two results classes.
            \begin{table}
                \caption{Streaming Caro-Wei Bound approximation.
                \emph{Arrival} is \textit{Edge} for edge-arrival or \textit{Vertex} for vertex-arrival. 
                \emph{Type} is \textit{Insert}\ for insertion-only or \textit{Turns}\ for turnstile. 
                \emph{Order} is \textit{Arb}\ for arbitrary or \textit{Ran}\ for random.
                \emph{Online} is ``Yes'' for \emph{online streaming} algorithm.
                Theorems marked as~`$^*$' have the max-degree constraint, which can be removed if an \emph{in-expectation guarantee} suffices. Constants~$c \in (0, 1)$ and $c' > 1$.
                }
                \label{table:ind-num-average-degree}
                \centering
                \begin{tabular}{|c|c|c|c|c|c|c|c|c|}
                \hline
                    \textbf{Arrival} & \textbf{Type} & \textbf{Order} & \textbf{Factor} & \textbf{(Work) Space} & \textbf{Online} & \textbf{Success Prob} & \textbf{Reference} \\
                \hline
                    Edge & Insert & Arb & 1 & $\mathcal{O}(n)$ & Yes & $\textit{Expected}$ &  \cite{halldorsson2016streaming} \\
                    Edge & Turns & Arb & $(1 \pm \eps)$ & $\mathcal{O}(\bar{d} \eps^{-2} \log n)$ & No & $c$ & \cite{cormode2018approximating} \\
                    Edge & Insert & Arb & $(1 \pm \eps)$ & $\mathcal{O}(\bar{d} \eps^{-2} \log n)$ & No & $c$ & Theorem~\ref{theo:caro-wei-size-degree-limit}$^*$ \\
                    Edge & Insert & Arb & $(1 \pm \eps)$ & $\mathcal{O}(\log \eps^{-1} \log n)$ & Yes & $c$ & Theorem~\ref{theo:caro-wei-online-streaming}$^*$ \\
                    Edge & Insert & Arb & $(1 \pm \eps)$ & $\mathcal{O}(\bar{d}^2 \eps^{-2} \log^2 n)$ & No & $n^{-c'}$ & Theorem~\ref{theo:ind-set-caro-wei-no-degree-constraint} \\
                    Vertex & Insert & Ran & $(1 \pm \eps)$ & $\mathcal{O}(\log (\bar{d} \eps^{-1}))$ & No & $c$ & Theorem~\ref{theo:caro-wei-size-vertex-arrival}$^*$ \\
                    Either & Either & Arb & $c$ & $\Omega({\bar{d}}/{c^2})$ & - & $c$ & \cite{cormode2018approximating} \\
                \hline
            \end{tabular}
            \end{table}

            \begin{table}[h]
                \caption{Estimating $\gamma, \beta$, and $\phi$ in streamed forests. All succeed with high probability.}
                \label{table:alg-tree-forest}
                \centering
                \begin{tabular}{|c||c|c|c|c|c|c|c|}
                \hline
                    \textbf{Problem}  & \textbf{Arrival} & \textbf{Type} & \textbf{Order} & \textbf{Pass} & \textbf{Factor} & \textbf{Space} & \textbf{Reference} \\
                \hline
                    $\gamma$ & Edge & Turns & Arb & 1 & $3 (1 \pm \eps)$ & $\log^{\mathcal{O}(1)} n$ & Theorem~\ref{theo:dom-3approx} \\
                    $\gamma$ & Edge & Turns & Arb & 2 & $2 (1 \pm \eps)$ & $\tilde{\mathcal{O}}(\sqrt{n})$ & Theorem~\ref{theo:dom-2approx} \\
                    $\gamma$ & Any & Any & Arb & 1 & $3/2 - \eps$ & $\Omega(\sqrt{n})$ & \cite{esfandiari2018streaming} \\
                    \hline
                    $\beta$ & Edge & Turns & Arb & 1 & $3/2 (1 \pm \eps)$ & $\log^{\mathcal{O}(1)} n$ & Theorem~\ref{theo:ind-3/2approx} \\
                    $\beta$ & Edge & Turns & Arb & 2 & $4/3 (1 \pm \eps)$ & $\tilde{\mathcal{O}}(\sqrt{n})$ & Theorem~\ref{theo:ind-4/3approx} \\
                    $\beta$ & Any & Any & Arb & 1 & $4/3 - \eps$ & $\Omega(\sqrt{n})$ & \cite{esfandiari2018streaming} \\
                    \hline
                    $\phi$ & Edge & Insert & Arb & 1 & $2 (1 \pm \eps)$ & $\tilde{\mathcal{O}}(\sqrt{n})$ & \cite{esfandiari2018streaming} \\
                    $\phi$ & Edge & Insert & Arb & 1 & $2 (1 \pm \eps)$ & $\log^{\mathcal{O}(1)} n$ & \cite{cormode2017sparse} \\
                    $\phi$ & Edge & Turns & Arb & 1 & $2 (1 \pm \eps)$ & $\log^{\mathcal{O}(1)} n$ & \cite{bury2015sublinear}, Theorem~\ref{theo:mat-2approx} \\
                    $\phi$ & Edge & Turns & Arb & 2 & $3/2 (1 \pm \eps)$ & $\tilde{\mathcal{O}}(\sqrt{n})$ & Theorem~\ref{theo:mat-3/2approx} \\
                    $\phi$ & Any & Any & Arb & 1 & $3/2 - \eps$ & $\Omega(\sqrt{n})$ & \cite{esfandiari2018streaming} \\
                \hline
            \end{tabular}
            \end{table}

                First, one-pass $\log^{\mathcal{O}(1)} n$-space algorithms, which relying on approximating both the numbers of leaves and non-leaf vertices in a stream. We have
                \begin{itemize}
                  \item a~$3/2\cdot (1 \pm \epsilon)$-approximation for~$\beta$;
                  \item a~$3\cdot (1 \pm \epsilon)$-approximation for~$\gamma$;
                  \item a~$2\cdot (1 \pm \epsilon)$-approximation  for~$\phi$.
                \end{itemize}
                
                Second, two-pass $\tilde{\mathcal{O}}(\sqrt{n})$-space algorithms\footnote{$\tilde{\mathcal{O}}$ suppresses the poly-logarithmic factor in the bound.}: 
                We further narrow the gap between the upper and lower bounds~\cite{esfandiari2018streaming}:
                \begin{itemize}
                    \item a~$4/3\cdot (1 \pm \epsilon)$-approximation algorithm for~$\beta$;
                    \item a~$2\cdot (1 \pm \epsilon)$-approximation algorithm for~$\gamma$;
                    \item a~$3/2\cdot  (1 \pm \epsilon)$-approximation algorithm for~$\phi$.
                \end{itemize}

                Our innovation is introducing the notion of support vertices from structural graph theory~\cite{delavina2010total}.
                A \emph{support vertex} is one that is adjacent to one or more leaves: an approximate count could be a key component of streaming graph algorithms.

\section{Preliminaries}

    Given a set of vertices in $G$,~$S$, let~$N(S)$ be the neighbors of~$S$ (excluding~$S$), and~$N[S]$ be~$S \cup N(S)$. Let~$\Delta$ be the max degree and~$\bar{d}$ be the average degree of $G$, and~$d(v)$ be the degree of vertex~$v$.~$\mathit{Deg}_{i}(G)$ is the set of vertices with~$d(v) = i$, while~$\mathit{Deg}_{\geq i}(G)$ is the set with~$d(v) \geq i$. A support vertex~\cite{delavina2010total} is a vertex adjacent to at least one leaf. We denote the set of support vertices by~$\supp(G)$.

    \subsection{Fundamental Streaming Algorithms}

        We assume vertices are~$[n] = \{1,2,\ldots,n\}$, and that~$n$ is known in advance.
        Second, our results are developed for certain graph classes, hence we assume that at the \emph{end} of the stream, the graph is guaranteed to be in the target class.

        We employ several streaming primitives. Consider a vector~$x \in \mathbb{R}^n$
        whose coordinates are updated by a turnstile stream with each update in the range~$[-M, M]$. 
        Let~$x_i$ be the final value of coordinate~$i$. 
        For~$p~>~0$, the~$L_p$ norm of~$x$ is~$\Vert x \Vert_p = (\sum_{i = 1}^n \lvert x_i \rvert^p)^{1/p}$. The~$L_0$ \emph{norm} is~$\Vert x \Vert_0 = (\sum_{i = 1}^n \lvert x_i \rvert^0)$. Both \emph{norms} can be~$(1 \pm \eps)$-approximated in streams in small space.

        \begin{theorem}
        \label{theo:Lp}
            \cite{kane2010exact} For~$p \in (0, 2)$ and~$\eps, \delta \in (0, 1)$, there is an algorithm that~$(1 \pm \eps)$-approximates~$\Vert x \Vert_p$ with probability at least~$1 - \delta$ and~$\mathcal{O}(\eps^{-2} \log (mM) \log \delta^{-1})$ space usage. Both update and reporting time are in~$\tilde{\mathcal{O}}(\eps^{-2} \log \delta^{-1})$.
        \end{theorem}

        \begin{theorem}
        \label{theo:L0}
            \cite{kane2010optimal} Given~$\eps, \delta \in (0, 1)$, there is an algorithm that, with probability at least~$1 - \delta$,~$(1 \pm \eps)$-approximates~$\Vert x \Vert_0$. It has an update and reporting time of~$\mathcal{O}(\log \delta^{-1})$, and uses space:
            $$\mathcal{O}(\eps^{-2} \log (\delta^{-1}) \log (n) (\log (1/\eps)~+~\log \log (mM))).$$
        \end{theorem}

        \noindent Given~$k \leq n$,~$x$ is~\emph{$k$-sparse} if it has at most~$k$ non-zero coordinates.

        \begin{theorem}
        \label{theo:k-sparse-recovery}
            \cite{cormode2014unifying} Given~$k$ and~$\delta \in (0, 1)$, there is an algorithm that recovers a~$k$-sparse vector exactly with probability~$1 - \delta$ and space usage~$\mathcal{O}(k \log n \cdot \log (k/\delta))$.
        \end{theorem}
        
        We focus on \emph{heavy hitters} that are coordinates of vector~$x$ with~$\lvert x_i \rvert \geq \eps \Vert x\Vert_1$.

        \begin{theorem}
        \label{theo:heavy-hitter-count-min}
            \cite{cormode2005improved} Given~$\eps, \tau \in (0, 1)$, there is an algorithm that outputs all items with frequency at least~$(\eps + \tau) \lVert x \rVert_1$, and with probability~$1 - \delta$ outputs no items with frequency less than~$\tau \lVert x \rVert_1$. It has update time~$\mathcal{O}(\log n \cdot \log ({2 \log n}/{(\delta \tau)}))$ and query time~$\mathcal{O}(\eps^{-1})$, and uses~$\mathcal{O}(\eps^{-1} \log n \cdot \log (2 \log n / (\delta \tau)))$ space.
        \end{theorem}

    \subsection{\updates{$\bf{\eps}$-min-wise Hash Functions}}

        To simulate uniform-at-random selection, we invoke carefully chosen hash families. A family of hash functions,~$\mathcal{H} = \{ h: [n] \rightarrow [m] \}$, is~``$\eps$-min-wise'' if for every~$A \subset [n]$ and~$x \in [n] \setminus A$,

        \begin{displaymath}
            \Pr_{h \in \mathcal{H}} [\forall a \in A, h(x) < h(a)] = \frac{1 \pm \eps}{\lvert A \rvert + 1}\,.
        \end{displaymath}
        \noindent
        Indyk~\cite{indyk2001small} showed that the~$\eps$-min-wise property can be achieved via a~$\mathcal{O}(\log \eps^{-1})$-wise hash family, under the constraints that~$\lvert A \rvert \in \mathcal{O}(\eps m)$. Its space usage is~$\mathcal{O}(\log \eps^{-1} \cdot \log m)$ and its computation time is~$\mathcal{O}(\log \eps^{-1})$.

    \subsection{Extension of Chernoff-Hoeffding inequality}

        Panconesi and Srinivasan \cite{panconesi1997randomized} extended the Chernoff-Hoeffding inequality to negatively correlated Boolean random variables as follows:

        \begin{theorem}
        \label{theorem:neg-chernoff}
            \cite{panconesi1997randomized} For~$r$ negatively correlated Boolean random variables~$X_1, \dots, X_r$, let~$X = \sum_{i = 1}^{r} X_i$,~$\mu = E[X]$ and~$0 < \delta < 1$, then

            \begin{displaymath}
                \Pr \big[ \lvert X - \mu \rvert \geq \delta \mu \big] \leq 2 e^{- \frac{\mu \delta^2}{2}}\,.
            \end{displaymath}
        \end{theorem}

\section{Estimating~$\beta$ via the Caro-Wei Bound}
\label{sec:caro-wei}

    In this section, we introduce efficient streaming algorithms to estimate the independence number,~$\beta$ in graphs with bounded average degree.
    A \emph{folklore} offline greedy algorithm returns an independent set whose size is in expectation the Caro-Wei Bound,~$\lambda$.
    Given a uniform-at-random permutation~$\pi$ of~$V$, the greedy algorithm adds vertex~$v$ to the solution whenever~$v$ is \emph{earlier} in~$\pi$ than all~$v$'s neighbors.
    Storing~$\pi$ explicitly requires~$\Omega(n \log n)$ bits, not sublinear, so disallowed here.

    We instead simulate a random permutation via function~$h$ drawn randomly from an~$\eps$-min-wise hash-family,~$\mathcal{H}$.
    Vertex~$v$ is added whenever its~$h(v)$ is less than all its neighbors' hash values.
    Family~$\mathcal{H}$ ensures the probability~$v$ has smallest hash is approximately proportional to~$\deg(v)$.
    As $\deg(v)$ could be~$\Theta(n)$, we rely instead on standard sampling and recovery techniques in our algorithm, \cwei, see Algorithm~\ref{alg:caro-wei-size-degree-limit}.
    The co-domain of~$h$ in \cwei\ is~$[n^3]$ for two reasons: (i) with high probability, there are no colliding pairs; (ii) we assume~$\eps \in \omega(n^{-2})$, so the~$\eps$-min-wise property applies to every subset of size~$\geq n$.
\begin{algorithm}
        \floatname{algorithm}{Algorithm}

        \caption{\cwei, estimating Caro-Wei bound}
        \label{alg:caro-wei-size-degree-limit}
        \begin{algorithmic}[1]

        \State{\textbf{Input}: Average degree,~$\bar{d}$,  error rate,~$\eps$.}

        \State{\textbf{Initialization}:
        $p \gets 4 (\bar{d} + 1)/(\eps^2 n)$}
        \State{$S \gets p n \text{ vertices sampled uniformly at random from } [n]$}
        \State{Select uniformly at random  $h \in \mathcal{H}$, $\eps$-min-wise hash functions: $[n] \rightarrow [n^3]$}

        \ForAll{$e = (u, v) \text{ in the stream}$}
            \If{$(u \in S) \wedge (h(u) \geq h(v))$}
                \State{Remove~$u$ from~$S$}
            \EndIf
            \If{$(v \in S) \wedge (h(v) \geq h(u))$}
                \State{Remove~$v$ from~$S$}
            \EndIf
        \EndFor

        \State{\textbf{return}~$\widehat{\lambda} \gets \lvert S \rvert / p$}

        \end{algorithmic}

    \end{algorithm}

    \begin{theorem}
    \label{theo:caro-wei-size-degree-limit}
        Given~$\eps \in (0, 1)$, if~$\Delta~\leq~{\eps^2~n}/({3~(\bar{d}+1)^3})$, in an insertion-only stream, 
        \cwei~$(1 \pm \eps)$-approximates~$\lambda$ w.p.~$2/3$ in work space~$\mathcal{O}(\bar{d}\, \eps^{-2} \log n)$.
    \end{theorem}

    \noindent
    To prove Theorem~\ref{theo:caro-wei-size-degree-limit}, we start with a result on the expectation of~$\widehat{\lambda}$.

    \begin{lemma}
    \label{lemma:caro-wei-size-degree-limit-exp}
       $E[\widehat{\lambda}] = (1 \pm \eps) \, \lambda$.
    \end{lemma}

    \begin{proof}
        Let binary random variables~$X_v$ denote whether vertex~$v$ is sampled in~$S$ and not removed during the stream. That is,~$X_v = 1$ if~$v \in S$ at the end of the stream. Let~$X = \sum_{v \in V(G)} X_v$, clearly,~$X = \lvert S \rvert$ at the end. It is not hard to see that~$X_v = 1$ if and only if~$v$ is sampled and~$v$ has the smallest hash value in~$N[v]$. According to the property of~$\eps$-min-wise hash families,~$v$ has the smallest hash value in~$N[v]$ with probability~$\frac{1 \pm \eps}{\lvert N[v] \rvert}$. By the construction of algorithm,~$v$ is sampled with probability~$p$. Since these two events are independent, the probability that~$X_v = 1$ is~$(1 \pm \eps) \, \frac{p}{d(v) + 1}$. Hence, by the linearity of expectation, we have
        \begin{displaymath}
            E[\lvert S \rvert] = \sum_{v \in V(G)} E[X_v] = p \sum_{v \in V(G)} \frac{1 \pm \eps}{d(v) + 1} = (1 \pm \eps) \, p \, \lambda\,.
        \end{displaymath}
        Hence,~$E[\widehat{\lambda}] = E[\frac{\lvert S \rvert}{p}] = (1 \pm \eps) \, \lambda$.
    \end{proof}
    We show that, with constant probability, estimate~$\widehat{\lambda}$ is no more than a~$1 \pm \eps$ factor away from its mean,~$E[\widehat{\lambda}]$.
    Let~$X_v$ be an indicator variable for the event that~$v$ is both sampled in, and not removed from,~$S$.
    Let $X = \sum_{v \in V} X_v$.
    Since some pairs of~$X_v$ and~$X_u$ are positively correlated, to estimate~$X$, we cannot directly apply a Chernoff(-like) bound. There are three cases to consider.
    \begin{itemize}
        \item Vertices~$u$ and~$v$ are adjacent:~$u \in N[v]$ and~$v \in N[u]$.
        If~$X_u$ is~$1$, wlog, we know that~$X_v$ is~$0$.
        Hence~$X_v$ and~$X_u$ are negatively correlated.

        \item Vertices~$u$ and~$v$ are not adjacent, but share at least one neighbor. Consider one such common neighbor,~$w$, i.e.,~$w \in N[u] \cap N[v]$.
        Let~$x < N(x)$ denote~$x$ having a smaller hash value than all elements in~$N(x)$.
        Knowing~$u < N(u)$ implies~$h(v) < h(w)$ is more likely,
        and hence~$v <N(v)$ is more likely, hence~$X_v$ and~$X_u$ are positively correlated, shown in
        Lemma~\ref{lemma:caro-wei-pos-correlation}. 

        \item All other cases:~$X_v$ and~$X_u$ are independent (see Lemma~\ref{lemma:caro-wei-pos-correlation}).
    \end{itemize}

    \begin{lemma}
    \label{lemma:caro-wei-pos-correlation}
        Let~$x$ and~$y$ be non-adjacent vertices with~$|N(x) \cap N(y)| = k$,~$|N(x) \setminus N(y)| = l$, and~$|N(y) \setminus N(x)| = r$. 
        From a uniform random permutation, 
        we have
            $\Pr[y < N(y) \mid x < N(x)] = {(l + r + 2k + 2)}/[(r + k + 1) (l + k + r + 2)]$.
        This also holds within a factor~$1 \pm \eps$ when the permutation is generated from a hash function from a $\eps$-min-wise family.
    \end{lemma}

    \begin{proof}    
        To begin with, we prove the statement with uniform random permutation. Consider~$n = l + k + r + 2$ vertices~$x, y$, $w_1, \ldots, w_l$, $u_1, \ldots, u_k$, and~$v_1, \ldots, v_r$.
        For simplicity, we use~$W$,~$U$, and~$V$ to denote~$\{w_1, \ldots, w_l\}$, $\{u_1, \ldots, u_k\}$, and~$\{v_1, \ldots, v_r\}$. Note that~$x$ has degree~$l + k$,~$y$ has degree~$k + r$,~$x$ and~$y$ are not adjacent and they share~$k$ neighbors, i.e.,~$N(x) = \{ W \cup U \}$ and~$N(y) = \{ U \cup V \}$.

        We use~$x < N(x)$ if~$x$ has the smallest order among~$N[x]$ in the permutation.
        Note that since we assume uniform random permutation, $\Pr[x < N(x)] = 1/(\lvert N(x) \rvert + 1) = 1/(l + k + 1)$. Consider the event $X < N(X)$ and $Y < N(Y)$, we have two cases:

        \begin{itemize}
            \item~$x$, followed by zero or more~$W$,~$y$, followed by a mixture of the rest of~$W$,~$V$, and~$U$.
            \item~$y$, followed by zero or more~$V$,~$x$, followed by a mixture of~$W$, the rest of~$V$, and~$U$.
        \end{itemize}

        By the uniform random permutation, the \emph{first} event occurs with probability

        \begin{displaymath}
            \frac{1}{l + k + r + 2} \cdot \frac{1}{k + r + 1} \, ,
        \end{displaymath}

        where the first factor is the probability of $x < \{y\} \cup W \cup V \cup U$, and the second factor is the probability of $y < V \cup U$, regardless of how it relates with vertices in $W$.
        Note these two events are independent of each other, since knowing~$x$ being the smallest across $y, W, U, V$ does not reveal anything about the ordering inside of $y, W, U, V$.
        By similar argument, the \emph{second} event occurs with probability
        \begin{displaymath}
            \frac{1}{l + k + r + 2} \cdot \frac{1}{l + k + 1}\,.
        \end{displaymath}
        Summing over the two cases, we have
        \begin{displaymath}
        \begin{aligned}
            & \Pr[y < N(y) \wedge x < N(x)] \\
            &= \frac{1}{l + k + r + 2} \cdot \frac{1}{k + r + 1}  + \frac{1}{l + k + r + 2} \cdot \frac{1}{l + k + 1} \\
            &= \frac{l + 2k + r + 2}{(l + k + r + 2)(k + r + 1)(l + k + 1)}
        \end{aligned}
        \end{displaymath}
        Finally, the conditional probability can be calculated as:
        \begin{equation}
        \begin{aligned}
        \label{eq:cond-prob}
            \Pr[y < N(y) \mid x < N(x)]
            &= \frac{\Pr[y < N(y) \wedge x < N(x)]}{\Pr[x < N(x)]} \\
            &= \frac{(l + r + 2k + 2)}{(l + k + r + 2) (k + r + 1)}\,.
        \end{aligned}
        \end{equation}
        In a permutation generated from a hash function from a $\eps$-min-wise family, the equation above also holds within a factor~$1 \pm \eps$: it relies only the probability of an element being the smallest.

        So far, we assume uniform random permutation, but an analogous claim holds with $\eps$-min-wise property, since we only consider the probability of an element being the smallest. The proof of $\eps$-min-wise property has similar structure, except we need to take $1 \pm \eps$ into account every time we calculate the probability of an element being smallest.
        We also observe that the statement ``knowing~$x$ being the smallest across $y, W, U, V$ does not reveal anything about the ordering inside of $y, W, U, V$'' requires a small amount of extra independence in the $\mathcal{O}(\log \eps^{-1})$-wise hash function.
        This has no effect on the asymptotic space bound for the hash function.
    
        Therefore, we have $\Pr[x < N(x)] = (1 \pm \epsilon)/(l + k + 1)$ and
    
        \begin{displaymath}
        \begin{aligned}
            & \Pr[y < N(y) \wedge x < N(x)] \\
            &= \frac{1 \pm \eps}{l + k + r + 2} \cdot \frac{1 \pm \eps}{k + r + 1}  + \frac{1 \pm \eps}{l + k + r + 2} \cdot \frac{1 \pm \eps}{l + k + 1} \\
            &= (1 \pm \eps)^2 \frac{l + 2k + r + 2}{(l + k + r + 2)(k + r + 1)(l + k + 1)}\,.
        \end{aligned}
        \end{displaymath}
    
        And, finally,
        \begin{equation}
        \begin{aligned}
        \label{eq:cond-prob-eps-min-wise}
            \Pr[y < N(y) \mid x < N(x)]
            &= \frac{\Pr[y < N(y) \wedge x < N(x)]}{\Pr[x < N(x)]} \\
            &= (1 \pm \eps) \frac{(l + r + 2k + 2)}{(l + k + r + 2) (k + r + 1)}\,.
        \end{aligned}
        \end{equation}
        Hence, if~$k = 0$, then~$\Pr[y < N(y) \mid x < N(x)] = {1}/(r + 1) = \Pr[y < N(y)]$. However, if~$k > 0$, then~$\Pr[y < N(y) \mid x < N(x)] > {1}/(r + k + 1)$. So~$x < N(x)$ and~$y < N(y)$ are positively correlated iff~$x$ and~$y$ have some common neighbor.
    \end{proof}
    
    With a Chernoff-type bound unavailable, we bound the variance of~$\widehat{\lambda}$, and apply Chebyshev's inequality:
    recall that $\text{Var}(\widehat{\lambda}) = \text{Var}({\lvert S \rvert}/{p}) = \text{Var}(X)$.

    \begin{lemma}
    \label{lemma:bound-exp-var}
        If~$\Delta \leq {\eps^2 n}/(3 (\bar{d} + 1)^3)$ and~$p \geq {4 (\bar{d} + 1)}/(\eps^2 n)$, then $\text{Var}(X) \leq {\eps^2 E^2[X]}/{3}$.
    \end{lemma}

    \begin{proof}
        We persist with the binary random variables from the proof of Lemma \ref{lemma:caro-wei-size-degree-limit-exp}. In the proof of Lemma \ref{lemma:caro-wei-size-degree-limit-exp}, we show that~$\Pr[X_v = 1] = \frac{(1 \pm \eps) p}{d(v) + 1}$. For simplicity, we ignore the~$(1 \pm \eps)$ factor in the following analysis, that is,~$\Pr[X_v = 1] = \frac{p}{d(v) + 1}$; the same analysis\footnote{Another convenient approach is to let~$p'$ be some value in~$[1-\eps,1+\eps]$ and work with~$p'$.} holds without removing~$(1 \pm \eps)$. For each~$X_v$, we have

        \begin{displaymath}
            \text{Var}(X_v) = E[X_v^2] - E^2[X_v] < \frac{p}{d(v) + 1}(1 - \frac{p}{d(v) + 1}) < \frac{p}{d(v) + 1}
        \end{displaymath}

        For non-adjacent vertices~$i$ and~$j$ share~$k \geq 0$ common neighbors, the covariance between~$X_i$ and~$X_j$ is

        \begin{equation}
        \begin{aligned}
        \label{eq:cov}
            & \text{Cov}(X_i, X_j \mid (i, j) \notin E) \\
            &= E[X_i X_j] - E[X_i]E[X_j] \\
            &= \Pr [X_i = 1 \wedge X_j = 1] - \Pr [X_i = 1] \Pr [X_j = 1] \\
            &= \Pr [X_i = 1 \mid X_j = 1] \Pr[X_j = 1] - \Pr [X_i = 1] \Pr [X_j = 1] \\
            &= ( \frac{p (d(i) + d(j) + 2)}{(d(i) + 1) (d(i) + d(j) - l + 2)} - \frac{p}{d(i) + 1} ) \frac{p}{d(j) + 1} \\
            &= \frac{p^2 l}{(d(i) + d(j) - l + 2) (d(i) + 1) (d(j) + 1)} \,,
        \end{aligned}
        \end{equation}

        where the second-last equality arises from Equation \eqref{eq:cond-prob}, where we set~$d(i) = l + k$ and~$d(j) = k + r$. It is not hard to see that Equation \eqref{eq:cov} is a constraint function subject to~$d(i) \geq 1$,~$d(j) \geq 1$, and~$k \leq \min \{ d(i), d(j)\}$. This function has global maximum at~$\frac{p^2 k}{(k + 1)^2 (k + 2)}$ when~$k = d(i) = d(j)$ for different values of~$d(i)$ and~$d(j)$, which is maximized at~$\frac{p^2}{12}$ when~$k = 1$. As indicated in Equation \eqref{eq:cond-prob-eps-min-wise}, when considering $\eps$-min-wise property instead of uniform random permutation, the global maximum is $\frac{(1 \pm \epsilon) p^2}{12} \leq \frac{p^2}{8}$ if $\eps < 1/2$.

        And for adjacent vertices~$i$ and~$j$,~$X_i$ and~$X_j$ are strongly negatively correlated, as~$X_i = 1$ implies~$X_j = 0$. Hence,

        \begin{displaymath}
        \begin{aligned}
            \text{Cov}(X_i, X_j \mid (i, j) \in E)
            = E[X_i X_j] - E[X_i]E[X_j]
            < 0
        \end{aligned}
        \end{displaymath}

        Let~$P$ be the number of vertex pairs that share at least one common neighbor. Assume the vertices are labeled from~$1$ to~$n$, by variance equation for the sum of correlated variables, we have (for simplicity, we sometimes abbreviate~$1 \leq i < j \leq n$ as~$i < j$),

        \begin{equation}
        \begin{aligned}
        \label{eq:var-1}
            \text{Var}(X)
            &= \sum_{i = 1}^n \text{Var}(X_i) + 2 \sum_{i < j} \text{Cov}(X_i, X_j) \\
            &< \sum_{i = 1}^n \text{Var}(X_i) + 2 \sum_{i < j \wedge (i, j) \notin E} \text{Cov}(X_i, X_j) \\
            &< \sum_{i = 1}^n \frac{p}{d(i) + 1} + 2 \sum_{i < j \wedge (i, j) \notin E} \frac{p^2 k}{(k + 1)^2 (k + 2)} \\
            &\leq p \, \lambda + 2 p^2 \frac{P}{8}
        \end{aligned}
        \end{equation}

        Note that each vertex~$v$ introduces at most~$\frac{d_v (d_v - 1)}{2} \leq \frac{d_v^2}{2}$ new pairs of vertices that share a common neighbor, i.e.,~$v$. Thus,~$P \leq \frac{1}{2} \sum_{v \in V(G)} d_v^2$. Since~$\Delta \leq \frac{\eps^2 n}{3 (\bar{d} + 1)^3}$, we have

        \begin{equation}
        \begin{aligned}
        \label{eq:var-2}
            \sum_{v \in V(G)} d_v^2
            \leq \Delta^2 \frac{\bar{d} n}{\Delta} + (n - \frac{\bar{d} n}{\Delta})
            &< (\Delta \bar{d} + 1) n \\
            &\leq (\frac{\eps^2 n \bar{d}}{3 (\bar{d} + 1)^3}  + 1) n \\
            &< \frac{\eps^2 n^2}{3 (\bar{d} + 1)^2} + n \\
            &\leq \frac{\eps^2 n^2}{2 (\bar{d} + 1)^2} \, ,
        \end{aligned}
        \end{equation}

        where the last inequality holds when~$n \geq \frac{\sqrt{6} (\bar{d} + 1)}{\eps}$. And the first inequality is because given the maximum degree~$\Delta$ and a sufficiently large~$n$, the term~$\sum_{v \in V(G)} d_v^2$ is maximized when the degree distribution is highly biased. That is, when there are around~$\frac{\bar{d} n}{\Delta}$ vertices with maximum degree~$\Delta$, while the rest of vertices have degree 1 (as the graph is assumed to be connected).

        Finally, combining Equation \eqref{eq:var-1} and \eqref{eq:var-2}, we have

        \begin{displaymath}
        \begin{aligned}
            \text{Var}(X)
            \leq p \, \lambda + 2 p^2 \frac{P}{8}
            &\leq p \, \lambda + p^2 \frac{\eps^2 n^2}{16 (\bar{d} + 1)^2} \\
            &\leq p \, \lambda + p^2 \frac{\eps^2 \lambda^2}{16} \\
            &= E[X] + \frac{\eps^2 E^2[X]}{16} \\
            &\leq \frac{\eps^2 E^2[X]}{3}
        \end{aligned}
        \end{displaymath}

        where the third inequality holds because the Tur{\'a}n Bound is at most the Caro-Wei Bound, i.e.,~$\frac{n}{\bar{d} + 1} \leq \sum_{v \in V(G)} \frac{1}{d(v) + 1} = \lambda$ (see \cite{boppana2018simple} for a proof). The inequality in the second-last line arises from~$E[X] = p \lambda$. And the last inequality holds if~$E[X] \geq \frac{16}{7 \eps^2}$; since we have assumed that~$p \geq \frac{4 (\bar{d} + 1)}{\eps^2 n}$ and~$\lambda \geq \frac{n}{\bar{d} + 1}$,~$E[X] = p \lambda \geq \frac{4}{\eps^2}$, the inequality follows.
    \end{proof}
    
    With Lemmas~\ref{lemma:caro-wei-size-degree-limit-exp}, \ref{lemma:caro-wei-pos-correlation} and~\ref{lemma:bound-exp-var}, we prove Theorem~\ref{theo:caro-wei-size-degree-limit}.
    
    \begin{proof}
        From Lemma~\ref{lemma:caro-wei-size-degree-limit-exp} and Lemma~\ref{lemma:bound-exp-var}, the expectation of the returned result,~$\widehat{\lambda}$, is~$(1 \pm \eps) \lambda$, and~$\text{Var}(X) \leq \frac{\eps^2 E^2[X]}{3}$. Hence, by Chebyshev's inequality,
        \begin{displaymath}
        \begin{aligned}
            \Pr \big[ \lvert \widehat{\lambda} - \lambda \rvert \geq 3 \eps \lambda \big]
            \leq
            \Pr \big[ \lvert \widehat{\lambda} - E[\widehat{\lambda}] \rvert \geq \eps E[\widehat{\lambda}] \big]
            =
            \Pr \big[ \lvert X - E[X] \rvert \geq \eps E[X] \big]
            \leq
            \frac{1}{3} \, ,
        \end{aligned}
        \end{displaymath}
        where the first inequality holds because~$\eps < 1$ so that~$(1 + \eps)^2 < (1 + 3 \eps)$ and~$(1 - \eps)^2 > (1 - 3 \eps)$. Therefore, the algorithm succeeds with probability at least~$2/3$.

        The algorithm stores an~$\eps$-min-wise hash function and~$\mathcal{O}(\bar{d} \eps^{-2})$ vertices. The hash function takes $\mathcal{O}(\log \eps^{-1} \cdot \log n)$ bits to store, and each vertex~$v$ takes~$\mathcal{O}(\log n)$ bits to store. Therefore, the total space usage is~$\mathcal{O}((\bar{d} \eps^{-2}  + \log \eps^{-1})\log n) = \mathcal{O}(\bar{d} \eps^{-2} \log n)$.
    \end{proof}
    
    Returning the median of several instances, the success probability of Algorithm~\ref{alg:caro-wei-size-degree-limit} becomes~$1 - \delta$ in~$\mathcal{O}(\bar{d} \eps^{-2} \log \delta^{-1} \cdot \log n)$ total space.
    
    The \cwei\ algorithm (Algorithm~\ref{alg:caro-wei-size-degree-limit}) can be further adapted to suit the following different problem settings.
    
    \begin{itemize}
        \item In the online streaming model, we are able to also output an actual independent set by allowing $n$ bits of external, solution-space memory. 
        \item The condition of having a bound on the maximum degree can be removed by excluding high-degree vertices, although this requires a post-processing stage and is thus not an online algorithm. 
        \item In a vertex-arrival stream, with vertices in \emph{random} order, the algorithm uses only $\mathcal{O}(\log(\bar{d} \eps^{-2}))$ space.
    \end{itemize}

    \subsection{Extension 1: An Independent Set in the Online Streaming model}
    \label{app:ext-1}

        Algorithm~\ref{alg:caro-wei-size-degree-limit} only returns an estimate of the Caro-Wei bound,~$\lambda$.
        It can be easily modified to give \cwei[1], which outputs an actual independent set in the \emph{online streaming} model.
        Select $h$ uniformly at random from the $\eps$-min-wise hash family $\mathcal{H}$.
        We initialise an $n$-bit array,~$A$, to true in the solution space.
        Upon each arriving edge $(u, v)$, ordered such that $h(u) \leq h(v)$, we irrevocably set the $v$-th index of~$A$ to false. At any time the set of indices set to true represents an independent set with respect to the edges observed so far. By the analysis above this set is indeed independent, its size has expected value in the range $(1 \pm \eps) \lambda$ and is concentrated around its expectation.
        The working space of our algorithm is merely the space to store the hash function, i.e.,~$\mathcal{O}(\log \eps^{-1} \cdot \log n)$.
        Moreover, the update time of our algorithm is the time to compute a hash value,~$\mathcal{O}(\log \eps^{-1})$.

        \begin{theorem}
        \label{theo:caro-wei-online-streaming}
            For a graph with average degree~$\bar{d}$ and max degree~$\Delta~\leq~{\eps^2~n}/({3~(\bar{d}~+~1)^3})$, \cwei[1] is an \emph{online streaming} algorithm that, with probability at least~$2/3$, reports an independent set of size~$1 \pm \eps$ times the Caro-Wei bound. The working space is~$\mathcal{O} (\log \eps^{-1} \cdot \log n)$ and update time is~$\mathcal{O}(\log \eps^{-1})$.
        \end{theorem}
        
        \begin{proof}
            Let a binary random variable~$X_v$ denote whether~$v$ is not removed by our algorithm. Since~$v$ is removed if and only if it does not have the smallest hash value across~$N[v]$,~$\Pr[X_v = 1] = \frac{1 \pm \eps}{d(v) + 1}$. Let~$X = \sum_{v \in V(G)} X_v$. By the construction of the algorithm,~$X$ is exactly the size of the final solution. And
            \begin{displaymath}
                E[X] = \sum_{v \in V(G)} \Pr[X_v = 1] = (1 \pm \eps) \sum_{v \in V(G)} \frac{1}{d(v) + 1} = (1 \pm \eps) \, \lambda\,.
            \end{displaymath}

            Following the same proof as Lemma~\ref{lemma:bound-exp-var} (note: since we are not sampling, the sampling probability~$p$ in Lemma \ref{lemma:bound-exp-var} can be regarded as 1), we can bound the variance of~$X$ as~$\text{Var}(X) \leq \frac{\eps^2 E^2[X]}{3}$. Applying Chebyshev's inequality, we have
            \begin{displaymath}
            \begin{aligned}
                \Pr \big[ \lvert \widehat{\lambda} - \lambda \rvert \geq 3 \eps \lambda \big]
                \leq
                \frac{1}{3} \, ,
            \end{aligned}
            \end{displaymath}
            Hence, the algorithm fails with probability at most~$1/3$.

            The algorithm only uses a randomly drawn~$\eps$-min-wise hash function, which can be stored using~$\mathcal{O} (\log \eps^{-1} \log n)$ bits and has calculation time~$\mathcal{O}(\log \eps^{-1})$. Hence, the space usage and the update time our algorithm is~$\mathcal{O} (\log \eps^{-1} \log n)$ and~$\mathcal{O}(\log \eps^{-1})$.
        \end{proof}

        The success probability can be boosted to~$1 - \delta$ by running~$\log \delta^{-1}$ independent instances simultaneously and return the maximum result. By the union bound, the probability that all instances fail is at most~$(\frac{1}{3})^{\log \delta^{-1}} = \delta$. After boosting, the new algorithm uses space~$\mathcal{O} (\log \eps^{-1} \log \delta^{-1} \log n)$ and has update time~$\mathcal{O}(\log \eps^{-1} \log \delta^{-1})$.

    \subsection{Extension 2: Removing the Constraint on Max Degree}
    \label{app:ext-2}

        The bound~$\frac{\eps^2~n}{3~(\bar{d}+1)^3}$ on the maximum degree,~$\Delta$, is a requirement for Theorem~\ref{theo:caro-wei-size-degree-limit}.
        Described as Subroutine~\ref{alg:ind-set-caro-wei-no-degree-constraint}, \cwei[2] avoids this constraint by first removing the heavy hitters (i.e., high-degree vertices) in the graph.
        In the Caro-Wei Bound,
        $$\sum_{v \in V(G)} \frac{1}{1 + \text{deg}(v)} \geq \lambda\,,$$ high-degree vertices have very little impact on the lower bound of~$\lambda$. We formalize this observation in Lemma~\ref{lemma:1-eps-caro-wei}.
        Let~$H_{\mu}$ be the set of vertices with degree at least~$\mu$ and~$G_{-H_{\mu}}$ be the induced subgraph on~$V \setminus H_{\mu}$.

        \begin{lemma}
        \label{lemma:1-eps-caro-wei}
            For every~$\mu \geq (\bar{d} + 1)/\sqrt{\eps}$, we have~$\lambda(G) \geq \lambda(G_{-H_{\mu}}) \geq (1 - \eps) \lambda(G)$.
        \end{lemma}
        
        \begin{proof}
            The upper bound,~$\lambda(G) \geq \lambda(G_{-\mathit{Deg}_{\geq \mu}})$, clearly holds as removing vertices does not increase~$\lambda$.

            Given the average degree~$\bar{d}$, by the Pigeonhole Principle, there are at most~$\frac{\bar{d} n}{\mu}$ vertices with degree at least~$\mu$. Hence, we have

            \begin{displaymath}
            \begin{aligned}
                \lambda(G_{-\mathit{Deg}_{\geq \mu}})
                > \lambda(G) - \sum_{v \in \mathit{Deg}_{\geq \mu}} \frac{1}{d(v) + 1}
                &\geq \lambda(G) - \frac{1}{\mu + 1} \frac{\bar{d} n}{\mu} \\
                &> \lambda(G) - \eps \frac{\bar{d} n}{(\bar{d} + 1)^2} \\
                &\geq \lambda(G) - \eps \lambda(G)\,,
            \end{aligned}
            \end{displaymath}
            where the last inequality holds as the Tur{\'a}n Bound is at most the Caro-Wei Bound, i.e.,~$\frac{n}{\bar{d} + 1} \leq \sum \frac{1}{d(v) + 1} = \lambda$ (see \cite{boppana2018simple} for a proof).
        \end{proof}

        Let~$R$ be the vertices to be ignored, which should%
        \footnote{Lemma~\ref{lemma:ind-set-caro-wei-no-degree-constraint} clarifies exactly which vertices are in~$R$.} be those with degree at least~$(\bar{d} + 1)/\sqrt{\eps}$.
        Ignoring~$R$, similar to Algorithm~\ref{alg:caro-wei-size-degree-limit}, Subroutine~\ref{alg:ind-set-caro-wei-no-degree-constraint} approximates~$\lambda$ in the remaining graph by sampling vertices and removing some, based on hash values.
        By Lemma~\ref{lemma:1-eps-caro-wei}, this estimate is at most factor~$1 - \eps$ away from the Caro-Wei bound of the original graph,~$G$.

        We identify the high-degree vertices,~$R$, via the heavy-hitter sketch, with specific parameter settings, as described in here.
        Since~$R$ is only reported at the end of stream,
        Subroutine~\ref{alg:ind-set-caro-wei-no-degree-constraint}
        retains all neighbors of vertices in~$S$, deferring the heavy-hitter ignoring and hash-value removal to post-processing.
        If two passes are allowed, or if~$R$ is provided in advance, then these two steps can be done on-the-fly, fitting into the online streaming model.
        In particular, the condition in line~6 of Algorithm~\ref{alg:caro-wei-size-degree-limit} becomes $(u \in S) \wedge (v \notin R) \wedge (h(u) \geq h(v))$.

        \begin{algorithm}
            \floatname{algorithm}{Subroutine}
            \caption{\cwei[2]: Estimating Caro-Wei Bound without Degree Constraint}

            \label{alg:ind-set-caro-wei-no-degree-constraint}
            \begin{minipage}{0.46\textwidth}
            \begin{algorithmic}[1]

            \State{\textbf{Input~1}: the average degree,~$\bar{d}$, the error rate,~$\eps$}

            \State{\textbf{Initialization}: Uniformly-at-random select~$h \in \mathcal{H}$, an $\eps$-min-wise hash family from~$[n]$ to~$[n^3]$}
            \State{$p \gets 4 (\bar{d} + 1)/(\eps^2 n)$}
            \State{$S \gets p n$ vertices sampled uniformly at random from $[n]$}

            \For{$u \in S$}
                \State{$L_u \gets \{\}$}
            \EndFor

            \ForAll{$e = (u, v)$ in the stream}
                \If{$u \in S$}
                    \State{Add~$v$ to~$L_u$}
                \EndIf
                \State{Perform the same operation on~$v$}
                \If{$\sum_{u \in S} (1 + \lvert L_u \rvert) > 10 \bar{d} p n$}
                \State{\textbf{Abort}}
                \EndIf
            \EndFor
            \end{algorithmic}

            \end{minipage}
            \begin{minipage}{0.46\textwidth}
                \begin{algorithmic}[1]
                \State{\textbf{Post-processing}:}
                \State{\textbf{Input~2}: a set of ignored vertices,~$R$}
                \ForAll{$v \in R$}
                    \ForAll{$u \in S$}
                        \State{$L_u \gets L_u \setminus \{v\}$}
                    \EndFor
                \EndFor

                \ForAll{$v \in S$}
                    \ForAll{$u \in L_v$}
                        \If{$h(u) < h(v)$}
                            \State{$S \gets S \setminus \{v\}$}
                        \EndIf
                    \EndFor
                \EndFor

                \State{\textbf{return}~$\widehat{\lambda} = \lvert S \rvert / p$}

                \end{algorithmic}
            \end{minipage}
        \end{algorithm}

        \begin{lemma}
        \label{lemma:ind-set-caro-wei-no-degree-constraint}
            Suppose~$R$ contains all vertices with degree at least~$\frac{\eps^2~n}{3~(\bar{d}+1)^3}$, but no vertices with degree less than~$(\bar{d} + 1)/{\sqrt{\eps}}$.
            Given~$\eps \in (0, 1)$, for every graph with average degree~$\bar{d}$,
            Subroutine~\ref{alg:ind-set-caro-wei-no-degree-constraint} is an insertion-only algorithm that~$(1 \pm \eps)$-approximates the Caro-Wei bound. Its work space is~$\mathcal{O}(\bar{d}^2 \eps^{-2} \log n)$; it succeeds with probability at least~$19/30$.
        \end{lemma}
        
        \begin{proof}
            By the construction of Subroutine \ref{alg:ind-set-caro-wei-no-degree-constraint}, vertices in~$R$ are removed from the neighborhood list of sampled vertices (line 15~--~17). Hence, by the condition on line 20, each sampled vertex is retained in~$S$ if and only if it has the smallest hash value across its neighbors in the subgraph after removing~$R$. Let~$X_v$ be the binary indicator of~$v$ being sampled and retained in~$S$. Since the sampling procedure and the hash function are independent,~$\Pr[X_v = 1] = (1 \pm \eps) \frac{p}{d'(v) + 1}$, where~$d'(v)$ is the degree of~$v$ after removing~$R$. Let~$X = \sum_{v \in V(G)} X_v$. Clearly,~$X = \lvert S \rvert$ at the end of stream. Let~$\lambda'$ be the Caro-Wei bound of the subgraph after removing~$R$, by the linearity expectation, the expectation of~$\widehat{\lambda'}$ is

            \begin{displaymath}
                E[\widehat{\lambda'}] = E[\lvert S \rvert] / p = \frac{1}{p} \, \sum_{v \in V(G)} E[X_v] = \sum_{v \in V(G)} \frac{ 1 \pm \eps}{d'(v) + 1} = (1 \pm \eps) \lambda' \, ,
            \end{displaymath}

            The variance of~$\widehat{\lambda'}$ is the same as the variance of~$X$, as~$p$ is fixed. By Lemma \ref{lemma:bound-exp-var}, condition on~$R$ contains all vertices with degree at least~$\frac{\eps^2~n}{3~(\bar{d}~+~1)^3}$, we have~$\text{Var}(X) \leq \frac{\eps^2 E^2[X]}{3}$. Hence,

            \begin{displaymath}
            \begin{aligned}
                    \Pr \big[ \lvert \widehat{\lambda'} - \lambda' \rvert \geq 3 \eps \lambda' \big]
                \leq
                \Pr \big[ \lvert X - E[X] \rvert \geq \eps E[X] \big]
                \leq
                \frac{\text{Var}(X)}{\eps^2 E^2[X]}
                \leq
                \frac{1}{3}\,.
            \end{aligned}
            \end{displaymath}

            By Lemma \ref{lemma:1-eps-caro-wei}, condition on~$R$ contains no vertices with degree smaller than
            ${\bar{d} (\bar{d} + 1)}/\eps$,~$\lambda' \geq (1 - \eps) \lambda$. Also, the algorithm aborts if the size of sampled vertices and their neighborhood lists are too large (line 11~--~12). Since the average degree is~$\bar{d}$, in expectation,~$\bar{d} p n$ vertices are stored by our algorithm. By Markov's inequality, the probability that the actual number of stored vertices is 10 times greater than its expectation,~$\bar{d} p n$, is at most~$1/10$. Applying the union bound, the probability that this algorithm fails or aborts is at most~$11/30$.

            The space usage of this algorithm is space to store the hash function, plus the space to store sampled vertices and their neighbors. Because of our constraints on line 11, there are at most~$\mathcal{O}(\bar{d} p n) = \mathcal{O}(\bar{d}^2 \eps^{-2})$ vertices being stored, where each vertex takes~$\mathcal{O}(\log n)$ bits to store. The hash function takes~$\mathcal{O}(\log \eps^{-1} \log n)$ bits to store. Hence, the total space usage is~$\mathcal{O}(\bar{d}^2 \eps^{-2} \log n)$.
        \end{proof}

        \subsubsection{\textbf{Identifying heavy hitters}}
        \label{sec:hh}

            The success of Subroutine~\ref{alg:ind-set-caro-wei-no-degree-constraint} depends on identifying the set of high-degree vertices,~$R$, well.
            By Theorem~\ref{theo:heavy-hitter-count-min}, running the algorithm of~Cormode and Muthukrishnan~\cite{cormode2005improved} with~$\psi = \tau = \frac{\eps^2}{6 (\bar{d} + 1)^4}$ and~$\delta = n^{-c'}$, for some constant~$c' > 1$, we obtain the desired removal set with high probability.
            This heavy-hitter technique requires~$\mathcal{O}(\eps^{-1} \log^2 n)$ space.
            By running~$c'\log n$ (for some constant~$c'$) instances of Subroutine~\ref{alg:ind-set-caro-wei-no-degree-constraint} concurrently, and returning the maximum result, the failure probability drops to~$n^{-c'}$.
            Importantly, there is only one instance of the heavy-hitter algorithm. The heavy-hitter algorithm is independent of Subroutine~\ref{alg:ind-set-caro-wei-no-degree-constraint}, but feeds its output~$R$ to every instance of Subroutine~\ref{alg:ind-set-caro-wei-no-degree-constraint}.
            Via a union bound, the whole process's failure probability is at most~$2 n^{-c'}$.
            \begin{theorem}
            \label{theo:ind-set-caro-wei-no-degree-constraint}
                Given~$\eps \in (0, 1)$, for every graph with average degree~$\bar{d}$, \cwei[2] is an algorithm that whp~$(1 \pm \eps)$-approximates the Caro-Wei bound in~$\mathcal{O}(\bar{d}^2 \eps^{-2} \log^2 n)$ space.
            \end{theorem}

    \subsection{Extension 3: Algorithm for Random-Order Vertex-Arrival Streams}
    \label{app:ext-3}

        So far, we have relied on a hash function drawn from an~$\eps$-min-wise hash family to approximate a random permutation.
        However, if the stream is vertex-arrival and random-order, the random permutation comes for free.
        Moreover, we can immediately tell whether the vertex is earliest in its neighborhood because the \emph{convention} is only to list prior-occurring neighbors.
        Hence, it suffices to count the number of such vertices, with space further reduced to~$\mathcal{O}(\log (\bar{d} \eps^{-1}))$
        by incrementing the counter with probability~$p = \frac{4 (\bar{d} + 1)}{\eps^2 n}$.
        We abort the algorithm if~$c > 10 p n$, which happens with probability less than~$1/10$.
        We conclude the following theorem:

        \begin{theorem}
        \label{theo:caro-wei-size-vertex-arrival}
            Given~$\eps \in (0, 1)$, for every graph with average degree~$\bar{d}$ and max degree~$\Delta~\leq~\frac{\eps^2~n}{3~(\bar{d}~+~1)^3}$, \cwei[3] is an insertion-only random-order vertex-arrival algorithm that~$(1 \pm \eps)$-approximates~$\lambda$. It succeeds with probability at least~$19/30$ in space~$\mathcal{O}(\log (\bar{d} \eps^{-2}))$.
        \end{theorem}
        
        \begin{proof}
            Firstly, note that~$\widehat{\lambda}$ is an unbiased estimator of~$\lambda$, as each vertex arrives uniformly at random. Let binary random variable~$X_v$ denote whether~$c$ is incremented when~$v$ arrives. That is,~$X_v = 1$ if the neighborhood list of~$v$ is empty, and the counter is increased. These two events are independent, the former event happens with probability~$\frac{1}{d(v) + 1}$, while the latter event happens with probability~$p$. Let~$X = \sum_{v \in V(G)} X_v$. By the linearity of expectation, we have
            \begin{displaymath}
                E[\widehat{\lambda}] = E[c] / p = E[X] / p = \frac{1}{p} \, \sum_{v \in V(G)} E[X_v] = \sum_{v \in V(G)} \frac{1}{d(v) + 1} = \lambda\,.
            \end{displaymath}

            By Lemma \ref{lemma:bound-exp-var}, the variance of~$X$,~$\text{Var}(X) \leq \frac{\eps^2 E^2[X]}{3}$. Applying Chebyshev's inequality, we have
            \begin{displaymath}
            \begin{aligned}
                \Pr \big[ \lvert \widehat{\lambda} - \lambda \rvert \geq \eps \lambda \big]
                =
                \Pr \big[ \lvert X - E[X] \rvert \geq \eps E[X] \big]
                \leq
                \frac{\text{Var}(X)}{\eps^2 E^2[X]}
                \leq
                \frac{1}{3} \, ,
            \end{aligned}
            \end{displaymath}

            The algorithm might fail if the counter,~$c$, becomes too large, i.e.,~$c > 10 p n$. By Markov's inequality,
            \begin{displaymath}
                \Pr[X \geq 10 p n] \leq \frac{E[X]}{10 p n} = \frac{\lambda}{10 n} \leq \frac{1}{10}\,.
            \end{displaymath}
            Therefore, applying the union bound, the algorithm succeeds with probability at least~$1 - 1/3 - 1/10 = 19/30$.

            The space usage of the algorithm is the size of the counter~$c$. By the construction of the algorithm,~$c$ is incremented to at most~$10 p n$. Hence, it can be stored in~$\mathcal{O}(\log (p n)) = \mathcal{O}(\log (\bar{d} \eps^{-2}))$ bits.
        \end{proof}

        Similarly, by running~$\log \delta^{-1}$ independent instances concurrently and return the maximum result, the success probability can be boosted to~$1 - \delta$. This is because the probability that all instances fail is at most~$(\frac{11}{30})^{\log \delta^{-1}} = \delta$. The space usage of the new algorithm after boosting is~$\mathcal{O}(\log (\bar{d} \eps^{-2}) \log \delta^{-1})$.

\section{Sublinear-space Streaming Algorithms for Approximating \\Graph Parameters in Forests}
\label{sec:forest}

    We now describe our sublinear-space streaming algorithms for approximating graph parameters when the input graph is a forest%
    \footnote{Unless otherwise stated, we assume that the input forest has no isolated vertices.}. Our algorithms work for turnstile edge-arrival streams. For each graph parameter, we design two algorithms: a $1$-pass algorithm using $\log^{\mathcal{O}(1)} n$ space, and a $2$-pass algorithm using $\tilde{\mathcal{O}}(\sqrt{n})$ space, with better approximation ratio (see Table~\ref{table:smaller-ub}).
    
    \begin{table}[H]
        \caption{Turnstile edge-arrival streaming algorithms for approximating $\gamma$, $\beta$, and $\phi$ in forests. Each succeeds with high probability.}
        \label{table:smaller-ub}
        \centering
        \begin{tabular}{|c||c|c|c|c|c|c|c|}
        \hline
            \textbf{Problem}  & \textbf{Number of Passes} & \textbf{Approximation Factor} & \textbf{Space} & \textbf{Reference} \\
            \hline
            $\beta$ & 1 & $3/2\cdot (1 \pm \eps)$ & $\log^{\mathcal{O}(1)} n$ & Theorem \ref{theo:ind-3/2approx} \\
            $\beta$ & 2 & $4/3\cdot (1 \pm \eps)$ & $\tilde{\mathcal{O}}(\sqrt{n})$ & Theorem \ref{theo:ind-4/3approx} \\
            \hline
            $\gamma$  & 1 & $3\cdot (1 \pm \eps)$ & $\log^{\mathcal{O}(1)} n$ & Theorem \ref{theo:dom-3approx} \\
            $\gamma$ & 2 & $2\cdot (1 \pm \eps)$ & $\tilde{\mathcal{O}}(\sqrt{n})$ & Theorem \ref{theo:dom-2approx} \\
            \hline
            $\phi$ & 1 & $2\cdot (1 \pm \eps)$ & $\log^{\mathcal{O}(1)} n$ & \cite{bury2015sublinear}, Theorem \ref{theo:mat-2approx} \\
            $\phi$  & 2 & $3/2\cdot (1 \pm \eps)$ & $\tilde{\mathcal{O}}(\sqrt{n})$ & Theorem \ref{theo:mat-3/2approx} \\
        \hline
    \end{tabular}
    \end{table}
    
\noindent 
    All algorithms can be described as two main steps:
    \begin{itemize}
        \item (Section~\ref{subsec:structural-bounds}) Obtain bounds on $\beta, \gamma, \phi$ in terms of other structural graph quantities such as number of leaves and number of support vertices, 
        \item (Section~\ref{sec:streaming-algorithms}) Estimate these structural graph quantities in sublinear-space.
    \end{itemize}

    \subsection{Structural Bounds on the Independence Number, Domination number and Matching number}
    \label{subsec:structural-bounds}

        In this section, we obtain (approximate) upper and lower bounds on the structural parameters $\beta, \gamma$ and $\phi$ using other structural parameters of forests. It is well-known that the number of \emph{leaves} and \emph{non-leaf vertices} can be used to approximate some of these graph parameters. For example, non-leaf vertices are used in~$2$-approximating the max matching~\cite{bury2015sublinear,esfandiari2018streaming} and $3$-approximating the min domination set in forests~\cite{eidenbenz2002online,lemanska2004lower,meierling2005lower}.
        They are helpful because one can always reason that some min dominating set must be a subset of them. 
        Similarly, there exists a max independent set that includes all leaves.

        Our main contribution here is to consider the notion of \emph{support vertices} and demonstrate their power in approximate the graph parameters $\beta, \gamma$ and $\phi$. We say that a vertex is a support vertex if it adjacent to at least one leaf. 
        Each of the graph parameters $\beta, \gamma$ and $\phi$  can be found in linear space via  standard dynamic programming and fixing the subset of leaves in the solution.
        Due to our (sublinear) space constraints, we cannot apply dynamic programming.
        However, going just one-level above the leaves -- counting the support vertices -- is sufficient to improve the approximations. For any forest $F$, we denote the set of its support vertices by $\supp(F)$. 

    \subsubsection{Structural Bounds on the Independence Number $\beta$}

        To start, we first show that any tree with at least $3$ vertices has a max independent set containing every leaf but no support vertex.
        \begin{lemma}
        \label{indNum_allLeaves}
            In every tree with~$n \geq 3$, there is a maximum independent set containing all leaves and no support vertices.
        \end{lemma}

        \begin{proof}

            This can be proved via adjusting vertices from a given maximum independent set. Given a maximum independent set, if it contains all leaves, then we are done. If some leaves are not included, then their support vertices must be in the independent set. This can be seen via contradiction, assume that both the support vertex and its leaves are not in the independent set, then clearly we could add all the leaves in the set, which violates the assumption that the set is maximum. Note that each support vertex must be adjacent to at least one leaf, and if it is in the independent set, none of its leaves is in the set. Hence, we could remove all support vertices from the given independent set, and add all their leaves into the set, the size of the resulting set does not decrease. Given that the original set is already maximum, the new set is also maximum.
        \end{proof}

        Hence, in the proof of Theorem \ref{theo:3/2approxIND-forest}, we assume that the max independent set of every tree contains every leaf unless~$n = 2$.

        \begin{theorem}
        \label{theo:3/2approxIND-forest}
            For every forest $F$ with $n$ vertices, we have $$\max \big\{\,\frac{n}{2} \,,\,\lvert \mathit{Deg}_{1}(F) \rvert - c \big\} \leq \beta(F) \leq \frac{1}{2} \big( n + \lvert \mathit{Deg}_{1}(F) \rvert \big) \,.$$
        \end{theorem}

        \begin{proof}

            To prove Theorem \ref{theo:3/2approxIND-forest}, it suffices to show that, for every tree~$T$ with~$|T| \geq 2$,
            \begin{displaymath}
                \max \Big\{ \frac{|T|}{2}, \lvert \mathit{Deg}_{1}(T) \rvert - 1 \Big\} \leq \beta(T) \leq \frac{|T| + \lvert \mathit{Deg}_{1}(T) \rvert}{2}\,.
            \end{displaymath}
            First we show the lower bound for $\beta(T)$. Since a tree is bipartite and each side of the bipartition forms an independent set, it follows that $\beta(T)\geq |T|/2$.
            Also, when~$|T| = 2$,~$\lvert \mathit{Deg}_{1}(T) \rvert = 2$ and $\beta(T)=1$. When~$|T| \geq 3$, no two leaves can be adjacent to each other and hence the set of all leaves forms an independent set, i.e., $\beta(T)\geq \lvert \mathit{Deg}_{1}(T) \rvert$. Summarizing, we have~$\beta \geq \max \big\{ \frac{|T|}{2}, \lvert \mathit{Deg}_{1}(T) \rvert - 1 \big\}$.

            Next, we show that~$\beta(T) \leq \frac{|T| + \lvert \mathit{Deg}_{1}(T) \rvert}{2}$ by induction on~$|T|$. We say that a pair of leaves of a tree are twins if they share a common neighbor (support vertex). We have two cases to consider depending on whether $T$ has twins or not:

            \paragraph*{Case 1:~$T$ has twins}
            Let the twins be~$v$,~$w$ and their shared neighbor be~$u$. Delete~$w$, and let the tree obtained be~$T' = T - w$. Then we have~$|T'| = |T| - 1$,~$\lvert \mathit{Deg}_1(T') \rvert = \lvert \mathit{Deg}_1(T) \rvert - 1$, and~$\beta(T') = \beta(T) - 1$. By induction hypothesis, we know that

            \begin{displaymath}
            \begin{aligned}
                & |T'| + |\mathit{Deg}_1(T')| \geq 2 \beta(T') \\
                & \Rightarrow
                |T| - 1) + \big(\lvert \mathit{Deg}_1(T) \rvert - 1\big) \geq 2\big(\beta(T) - 1\big) \\
                & \Rightarrow
                |T| + \lvert \mathit{Deg}_1(T) \rvert \geq 2 \beta(T)\,.
            \end{aligned}
            \end{displaymath}

            \paragraph*{Case 2:~$T$ does not have twins} 
            Let a leaf~$w$ have a parent~$u$ whose parent is~$x$. We can assume that~$u$ has no other leaves adjacent to it since~$T$ has no twins. Delete both~$u$ and~$w$, and let the obtained tree be~$T' = T - \{u,w\}$. It is easy to see that~$|T'| = |T| - 2$ and~$\lvert \mathit{Deg}_1(T) \rvert \geq \lvert \mathit{Deg}_1(T') \rvert \geq \lvert \mathit{Deg}_1(T) \rvert - 1$, because~$x$ may or may not be a leaf in~$T'$ but~$w$ was definitely a leaf in~$T$. Given an independent set~$I$ in $T$, we can obtain an independent set~$I' = I - w$ in~$T'$. Similarly, given an independent set~$H'$ in $T'$, we can obtain an independent set~$H = H' \cup \{ w \}$ in~$T$. Hence, it follows that~$\beta(T') = \beta(T) - 1$. By induction hypothesis, we have

            \begin{displaymath}
            \begin{aligned}
                & |T'| + \lvert \mathit{Deg}_1(T') \rvert \geq 2 \beta(T') \\
                & \Rightarrow
                \big(|T| - 2\big) + (\lvert \mathit{Deg}_1(T') \rvert) \geq 2(\beta(T) - 1) \\
                & \Rightarrow
                |T| + \lvert \mathit{Deg}_1(T) \rvert \geq 2\beta(T)\,.
            \end{aligned}
            \end{displaymath}
        \end{proof}
        Recall that for a forest $F$ we denote its set of leaves by $\mathit{Deg}_{1}(F)$ and its set of support vertices by~$\supp(F)$.

        \begin{theorem}
        \label{theo:4/3approxIND}
            Let $F$ be a forest with $n$ vertices. Then we have $$\frac{1}{2} \big( n + \lvert \mathit{Deg}_{1}(F) \rvert - \lvert \supp(F) \rvert \big) \leq \beta(F) \leq \frac{2}{3} \big( n + \lvert \mathit{Deg}_{1}(F) \rvert - \lvert \supp(F) \rvert \big) \,.$$
        \end{theorem}

        \begin{proof}
            By Lemma \ref{indNum_allLeaves}, we assume that (unless~$n = 2$)  the maximum independent set contains every leaf and none of the support vertices. We first show a lower bound (Proposition~\ref{indNum_Tree_LowerBound}) and upper bound (Proposition~\ref{indNum_Tree_UpperBound}) on the independence number of trees.  Theorem~\ref{theo:4/3approxIND} then follows immediately by applying Proposition~\ref{indNum_Tree_LowerBound} and Proposition~\ref{indNum_Tree_UpperBound} to each tree component of the forest~$F$.
        \end{proof}

        \begin{proposition}
        \label{indNum_Tree_LowerBound}
            For every tree~$T$ with~$|T| \geq 2$,

            \begin{displaymath}
                \beta \geq \frac{1}{2} (|T| + \lvert \mathit{Deg}_{1}(T) \rvert - \lvert \supp(T) \rvert)\,.
            \end{displaymath}
        \end{proposition}

        \begin{proof}
            We prove the proposition by showing that there exists an independent set of size at least~$\frac{1}{2} \big(|T| + \lvert \mathit{Deg}_{1}(T) \rvert - \lvert \supp(T) \rvert\big)$. When~$|T| = 2$, then $T$ is a path on two vertices which implies ~$\beta(T) = 1$ and~$\lvert \mathit{Deg}_1 \rvert = 2= \lvert S \rvert = 2$ and the desired inequality holds. Hence, we assume~$n > 2$.

            It suffices to show that there exists an independent set $H$ containing all leaves and half of the~$\mathit{Deg}_{\geq 2}(T)$ vertices that are not support vertices. We first remove all leaves and support vertices from $T$ to obtain a forest $F'$ with~$\lvert \mathit{Deg}_{\geq 2}(T) \rvert - \lvert \supp(T) \rvert$ vertices. Since each forest is bipartite and each side of a bipartition forms an independent set, it follows that $F'$ has an independent set $H'$ of size $\geq |F'|/2$. We take $H$ to be the union of $\mathit{Deg}_{1}(T)$ and $H'$: it follows that $H$ is an independent set in $T$ since $F'$ does not contain any support vertices of $T$. The size of $H$ is $|\mathit{Deg}_{1}(T)| + \frac{|\mathit{Deg}_{\geq 2}(T) -  \supp(T)|}{2} = \frac{|T|+ |\mathit{Deg}_{1}(T)| - \supp(T)}{2}$ since $|T|=|\mathit{Deg}_{1}(T)| + |\mathit{Deg}_{\geq 2}(T)|$.
            This concludes the proof of Proposition~\ref{indNum_Tree_LowerBound}.
        \end{proof}

        The lower bound of Proposition~\ref{indNum_Tree_LowerBound} is tight: consider a path $P$ on $2n$ vertices. Then $\beta(P)=n$ and~$\lvert \mathit{Deg}_{1}(P) \rvert = 2 = \lvert \supp(P) \rvert$.

        \begin{proposition}
        \label{indNum_Tree_UpperBound}
            For every tree~$T$ with~$|T| \geq 2$,

            \begin{displaymath}
                \beta(T) \leq \frac{2}{3} (|T| + \lvert \mathit{Deg}_{1}(T) \rvert - \lvert \supp(T) \rvert)\,.
            \end{displaymath}
        \end{proposition}

        \begin{proof}
            We prove this lemma via induction on the graph order,~$n$. Consider all trees with~$2 \leq n \leq 4$ as the base cases: these can all be easily verified by hand. Assume the inequality holds for every tree with order~$n - 1$.
            There are two cases to consider depending on whether $T$ has twins or not:

            \paragraph*{Case 1: There exist twins in $T$}
            Let~$T'$ be the tree after removing either one of the twins. Then we have $|T| = |T'| + 1$,~$\lvert \mathit{Deg}_1(T) \rvert = \lvert \mathit{Deg}_1(T') \rvert + 1$, and~$\lvert \supp(T) \rvert = \lvert \supp(T') \rvert$. According to Lemma \ref{indNum_allLeaves}, we have that~$\beta(T) = \beta(T') + 1$. Hence, by the induction hypothesis, the inequality holds.

            \paragraph*{Case 2: There are no twins in $T$}
            Consider the longest path in~$T$, denote one of its leaves as~$x$, the parent of~$x$ as~$y$, the parent of~$y$ as~$w$, and the parent of~$w$ as~$z$. If~$T$ has more than four nodes and no twins, then~$z$ is guaranteed to be non-leaf. This can be seen by contradiction. Assuming~$z$ is a leaf and there is more than four nodes. By our choice of the longest path, the neighbors of~$w$ and~$y$ can only be leaves. Since there is no twin,~$w$ and~$y$ both have degree two, there are exactly four nodes, a contradiction. By similar argument, it is not hard to see that~$y$ must have degree two.

            \paragraph*{Case 2.1:~$d(w) > 2$} 
            Remove both~$x$ and~$y$, and let the resulting tree be~$T'$. Note that this removes a leaf and a support vertex. Because~$w$ has degree more than 2, the removal of~$x$ and~$y$ does not make~$w$ a leaf, and~$w$ is a support vertex in~$T'$ if and only if it is a support vertex in~$T$. Therefore, we have ~$|T| = |T'| + 2$,~$\lvert \mathit{Deg}_1(T) \rvert = \lvert \mathit{Deg}_1(T') \rvert + 1$, and~$\lvert \supp(T) \rvert = \lvert \supp(T') \rvert + 1$. Moreover, we claim that~$\beta(T) = \beta(T') + 1$. This is because by Lemma~\ref{indNum_allLeaves}, there is a max independent set $H$ such that ~$x\in H$ and~$y\notin H$: thus the existence of~$x$ and~$y$ has no impact on whether~$w$ and the rest of vertices are in~$\beta(T)$ or not, i.e., they are in~$\beta(T)$ if and only if they are also in~$\beta(T')$. By the induction hypothesis, the inequality holds.

            \paragraph*{Case 2.2:~$d(w) = 2$}
            We have two cases to consider.

            \paragraph*{Case 2.2.1:~$z \notin \supp(T)$}
            We remove~$x$ and~$y$. Similar to Case 2.1, we have~$|T| = |T'| + 2$ and~$\beta(T) = \beta(T') + 1$. However, since~$d(w) = 2$ and~$z$ is not a support vertex in~$T$, the removal of~$x$ and~$y$ makes~$w$ a leaf and~$z$ a support vertex. Hence~$\lvert \mathit{Deg}_1(T) \rvert = \lvert \mathit{Deg}_1(T') \rvert$, and~$\lvert \supp(T) \rvert = \lvert \supp(T') \rvert$. By the induction hypothesis, the inequality holds.

            \paragraph*{Case 2.2.2:~$z \in \supp(T)$} 
            Remove~$x$,~$y$, and~$w$ from~$T$ and denote the remaining graph as~$T'$. Since~$z$ is a support vertex, according to Lemma \ref{indNum_allLeaves}, there is a max indpendent set $H$ of $T$ such that $z\notin H$ which implies that~$w$ must in the maximum independent set. Therefore, we have~$|T| = |T'| + 3$ and~$\beta(T) = \beta(T') + 2$. Moreover, removing these vertices make~$z$ neither a leaf nor a support vertex in~$T'$, unless~$d(z) = 2$. In the former case, we have~$\lvert \mathit{Deg}_1(T) \rvert = \lvert \mathit{Deg}_1(T') \rvert + 1$, and~$\lvert \supp(T) \rvert = \lvert \supp(T') \rvert + 1$. By the induction hypothesis, the inequality holds. In the latter case, $T$ is a path of length 5: the claim can easily verified by hand for this specific graph.
            This concludes the proof of Proposition~\ref{indNum_Tree_UpperBound}.
        \end{proof}

        This upper bound of Proposition~\ref{indNum_Tree_UpperBound} is asymptotically tight: consider the graph constructed from a star on ~$(r+1)$ vertices by replace each of the $r$ leaves with a path consisting of 3 vertices. Clearly, this graph is a tree, say $T$, with~$3r + 1$ vertices, and~$\lvert \supp(T) \rvert = \lvert \mathit{Deg}_{1}(T) \rvert = r$. Also,~$\beta(T) = 2 r$ as all the~$r$ leaves and their~$r$ grandparents form a maximum independent set.

        \noindent
        Furthermore, combining Theorem \ref{theo:4/3approxIND} with Theorem \ref{theo:3/2approxIND-forest}, we have:

        \begin{corollary}
        \label{coro:4/3approxIND-noS}
            If~$\lvert \supp \rvert \leq \frac{1}{2} \, \lvert \mathit{Deg}_{\geq 2} \rvert$, $\frac{3}{8} \big( n + \lvert \mathit{Deg}_{1} \rvert \big) \leq \beta \leq \frac{1}{2} \big( n + \lvert \mathit{Deg}_{1} \rvert \big) \,$.
        \end{corollary}

        \begin{proof}
            By definition,~$n = |F|=\lvert \mathit{Deg}_{\geq 2}(F) \rvert + \lvert \mathit{Deg}_{1}(F) \rvert$ and~$\lvert \supp(F) \rvert \leq \lvert \mathit{Deg}_{1}(F) \rvert$. Hence, if~$\lvert \supp(F) \rvert \leq \frac{\lvert \mathit{Deg}_{\geq 2}(F) \rvert}{2}$, we have the following inequality,

            \begin{equation}
            \label{temp_1}
                \lvert \supp(F) \rvert
                \leq \frac{\lvert \mathit{Deg}_{\geq 2}(F) \rvert}{4} + \frac{\lvert \mathit{Deg}_{1}(F) \rvert}{2}
                = \frac{n + \lvert \mathit{Deg}_1 \rvert}{4}\,.
            \end{equation}
            Utilizing Equation \eqref{temp_1} and combining Theorem \ref{theo:3/2approxIND-forest} and Theorem \ref{theo:4/3approxIND}, if~$\lvert \supp(F) \rvert \leq \frac{\lvert \mathit{Deg}_{\geq 2}(F) \rvert}{2}$, then
            \begin{displaymath}
                \frac{3 (n + \lvert \mathit{Deg}_{1}(F) \rvert)}{8} \leq \frac{n + \lvert \mathit{Deg}_{1}(F) \rvert - \lvert \supp(F) \rvert}{2} \leq \beta(F)\,.
            \end{displaymath}
            By Theorem \ref{theo:3/2approxIND-forest}, we have $\beta(F) \leq \frac{n + \lvert \mathit{Deg}_{1}(F) \rvert}{2}$ and hence Corollary \ref{coro:4/3approxIND-noS} follows.
        \end{proof}

        \subsubsection{Structural Bounds on the Domination Number $\gamma$}
        
            Many prior works demonstrated that domination number $\gamma$ can be $3$-approximated using number of non-leaf vertices (vertices with degree at least two) and number of tree components $c$.

            \begin{theorem}
            \label{theo:3approxDom-forest}
                \cite{eidenbenz2002online,meierling2005lower,lemanska2004lower}
                For every forest~$F$ with $c$ components we have $$\frac{1}{3} \, \lvert \mathit{Deg}_{\geq 2}(F) \rvert \leq \gamma(F) \leq \lvert \mathit{Deg}_{\geq 2}(F) \rvert + c\,.$$
            \end{theorem}
            
            With the help of support vertices and following a similar approach, we show that the approximation ratios above can be further improved to $2$.

            \begin{theorem}
            \label{theo:dom-num_2approx}
                For every forest $F$ we have $\frac{1}{4} \big( \lvert \mathit{Deg}_{\geq 2}(F) \rvert + \lvert \supp(F) \rvert \big) \leq \gamma \leq \frac{1}{2} \big( \lvert \mathit{Deg}_{\geq 2}(F) \rvert + \lvert \supp(F) \rvert \big).$
            \end{theorem}

            \begin{proof}
                By Lemma \ref{lemma:suppVertex}, we will assume that (unless~$n < 3$) the minimum dominating set contains all support vertices and no leaves. We first show a upper bound (Proposition~\ref{lemma:graphUpperBound}) and lower bound (Proposition~\ref{lemma:treeLowerBound}) on the domination number of trees.  Theorem~\ref{theo:dom-num_2approx} then follows immediately by applying Proposition~\ref{lemma:graphUpperBound} and Proposition~\ref{lemma:treeLowerBound} to each tree component of the forest $F$.
            \end{proof}
    
            \begin{lemma}
            \label{lemma:suppVertex}
                In every connected graph with~$n \geq 3$, there is a minimum dominating set that contains all support vertices and no leaves.
            \end{lemma}
    
            \begin{proof}
    
                We prove this by showing that every minimum dominating set can be adjusted to contain no leaves and all support vertices without increasing its size. Clearly, in every connected graph with~$n \geq 3$, all support vertices have degree greater than 1 (i.e., are not themselves leaves). If a support vertex has degree 1, then its neighbor must be a leaf, hence~$n = 2$, violating our assumption about~$n \geq 3$. Given a minimum dominating set, If it contains no leaves and all support vertices, then we are done. If not, we can remove the leaves from the dominating set and add their neighbors into the set. By the definition of support vertices, the neighbors of leaves are themselves support vertices. The resulting set is still a dominating set as the newly added support vertices dominate the removed leaves and themselves. Moreover, before the adjustment, for each pair of support vertex and leaf, at least one of them is in the dominating set, otherwise the domination property is broke. The adjustment ensures that the new dominating set contains all support vertices and no leaves. The size of the new dominating set is at most the size of original dominating set, because each removed leaf introduces at most one neighbor (i.e., support vertex) into the new set.
            \end{proof}

            Hence, in the following section, we assume the minimum dominating set contains all support vertices and no leaves. Now we give our upper and lower bounds on the domination number of forest. For the upper bound, we have
    
            \begin{proposition}
            \label{lemma:graphUpperBound}
                For every forest~$F$, we have
    
                \begin{displaymath}
                    \frac{\lvert \mathit{Deg}_{\geq 2}(F) \rvert + \lvert \supp(F) \rvert}{2} \geq \gamma(F)\,.
                \end{displaymath}
            \end{proposition}
    
            \begin{proof}
    
                If $|F| = 2$, then $F$ has to be an edge. Every vertex of $F$ is both a support vertex and a leaf, and we can pick either of them as the dominating set. Hence,
                $$\frac{\lvert \mathit{Deg}_{\geq 2}(F) \rvert + \lvert \supp(F) \rvert}{2} = \frac{1+1}{2} = 1 = \gamma(F)\,,$$ the inequality holds.
    
                Next, we prove the inequality for graphs with~$n > 2$. It is known that if a graph~$G$ is split into~$k$ disjoint subgraphs,~$G_1, G_2, \dots , G_k$, we have~$\gamma(G) \leq \gamma(G_1) + \gamma(G_2) + \dots + \gamma(G_k)$. This is because connecting two or more disjoint graphs only introduces new edges, thus vertices do not become undominated. Rewriting the inequality LHS as~$$\lvert \supp(F) \rvert + \frac{\lvert \mathit{Deg}_{\geq 2}(F) \rvert - \lvert \supp(F) \rvert}{2}\,,$$
                the first term could be viewed as adding all vertices in~$\supp(F)$ into the dominating set. By doing so, all of the vertices in~$N\big[\supp(F)\big]$ are dominated (note that this includes all the leaves). The second term could be viewed as an upper bound on the domination number after removing all the leaves and support vertices. Clearly, the remaining graph is a forest with~$(\lvert \mathit{Deg}_{\geq 2}(F) \rvert - \lvert \supp(F) \rvert)$ vertices. Thus it remains to prove that in such a forest, there exists a dominating set of size at most~$(\lvert \mathit{Deg}_{\geq 2}(F) \rvert - \lvert \supp(F) \rvert)/{2}$.
    
                Note that all isolated vertices in the induced graph are already dominated by some vertices in~$S$, because a vertex becomes isolated only if all of its neighbors are in~$S$. Thus, we only need to show that for each tree components (of size~$k > 1$) in the remaining forest, there exists a dominating set of size at most~${k}/{2}$. Ore \cite{ore1962theory} proved that for every graph~$G$ with order~$n$ and no isolated vertices,~$\gamma(G) \leq {n}/{2}$. Hence, this proposition follows.
            \end{proof}
    
            The upper bound of Proposition~\ref{lemma:graphUpperBound} is tight: construct a tree $T$ by taking a path of on $r$ vertices for any $r\geq 2$, add a leaf to every vertex on the path. It is easy to see that~$|\supp(T)| = |\mathit{Deg}_{\geq 2}(T)| = \gamma(T) = r$.
    
            Next, we present our lower bound result.
    
            \begin{proposition}
            \label{lemma:treeLowerBound}
                For every tree~$T$ with~$T| \geq 2$,
                \begin{displaymath}
                \gamma(T) \geq \frac{\lvert \mathit{Deg}_{\geq 2}(T) \rvert + \lvert \supp(T) \rvert}{4}\,.
                \end{displaymath}
            \end{proposition}
    
            \begin{proof}
                We prove the proposition by induction on~$n=|T|$. Consider the base cases as trees with~$n \leq 4$, which can be easily verified by hand.
    
                It remains to show the inductive step. Two or more leaves are twins if they share a common parent node (i.e., support vertex). If the tree contains twins, we remove any one of them from $T$ to obtain a tree $T'$. By Lemma~\ref{lemma:suppVertex}, the minimum dominating set of $T$ does not contain either of the twins. Hence, we have that ~$\gamma(T)=\gamma(T')$,~$\lvert \mathit{Deg}_{\geq 2}(T) \rvert=|\mathit{Deg}_{\geq 2}(T')|$, and~$\lvert \supp(T)=\supp(T') \rvert$ stay the same. By induction hypothesis, Proposition \ref{lemma:treeLowerBound} holds.
    
                Next, assume there are no twins and consider the longest path, say~$P$, in $T$. Denote the two endpoints as~$x$ and~$x'$ respectively, let~$y$ be the parent of~$x$,~$w$ be the parent of~$y$, and~$z$ be the parent of~$w$. Note that the length of~$P$ must be greater than 4, thus~$z$ and~$x'$ must be different. This can be seen via contradiction. Assuming~$P$ has length~$4$, by our choice of longest path~$P$,~$x$ and~$z$ (or equivalently,~$x'$) must be leaves, and the neighbors of~$y$ and~$w$ must be also leaves. Otherwise, our choice of the longest path is violated. If~$y$ or~$w$ is adjacent to leaves other than~$x$ and~$z$, then there exist twins in the graph, violating our assumption of no twins. Hence, the graph contains exactly 4 vertices, $x$, $y$, $w$, and $z$. It should be considered as the base case. Similar arguments can be obtained for~$P$ with length less than~$4$. Hence,~$P$'s length must be greater than 4.
    
                If~$w$ is a support vertex, then we are done by induction. By Lemma \ref{lemma:suppVertex}, the minimum dominating set contains both~$w$ and~$y$ as they are support vertices. Note that~$y$ must have degree exactly 2, otherwise the graph has either a twin, or a path longer than~$P$. Delete $x$ and let the tree obtained be $T''$. Then~$y$ becomes a leaf which is already dominated by~$w$. Therefore, after removing~$x$, we have~$\gamma(T'')=\gamma(T)-1$. Since~$y$ becomes a leaf, we have~$|\mathit{Deg}_{\geq 2}(T'')|=|\mathit{Deg}_{\geq 2}(T)|-1$ and~$|\supp(T'')|=|\supp(T)|-1$, and hence the inequality holds by our induction hypothesis.
    
                It remains to show that~$\gamma(T) \geq (\lvert \mathit{Deg}_{\geq 2}(T) \rvert + \lvert \supp(T) \rvert)/{4}$ when there are no twins and~$w$ is not a support vertex. Let~$T'$ and~$T''$ be the two sub-trees formed by deleting the edge~$(w,z)$, such that~$T'$ contains~$w$ and~$T''$ contains~$z$. Note that by our previous analysis, only paths of length 2 can be incident on~$w$, otherwise the tree falls into previous cases. Denote the number of such paths as~$r$; clearly,~$r \geq 1$ because there is path from~$x$ to~$w$. In tree~$T'$, clearly, by picking all~$r$ child nodes of~$w$, we obtain a minimum dominating set. Since~$w$ is not included in the set,~$\gamma(T) = \gamma(T') + \gamma(T'') = r + \gamma(T'')$. Let~$\psi(T) = \lvert \mathit{Deg}_{\geq 2}(T) \rvert + \lvert \supp(T) \rvert$. We have the following three inequalities:
    
                \begin{enumerate}
                    \item if~$d(z) > 1$ in~$T''$:~$\mathit{Deg}_{\geq 2}(T) \leq \mathit{Deg}_{\geq 2}(T'') + (r + 1)$.
                    \item if~$d(z) = 1$ in~$T''$:~$\mathit{Deg}_{\geq 2}(T) \leq \mathit{Deg}_{\geq 2}(T'') + (r + 1) + 1$.
                    \item~$\supp(T) \leq \supp(T'') + r$.
                \end{enumerate}
    
                The first two inequalities hold because~$T'$ contains~$r$ vertices from $\mathit{Deg}_{\geq 2}(T')$ vertices, and by adding the edge~$(w, z)$ we have~$d(w) \geq 2$. Also, if~$z$ has degree 1 in~$T''$,~$z$ is not in~$\mathit{Deg}_{\geq 2}(T'')$, but adding the edge~$(w, z)$ puts it in~$\mathit{Deg}_{\geq 2}(T)$. Otherwise,~$z$ is already in~$\mathit{Deg}_{\geq 2}(T'')$. Similarly, the \emph{third} inequality holds because there are~$r$ vertices in~$\supp(T')$, adding the edge~$(w, z)$ does not introduce new support vertices in~$\supp(T)$. Hence,~$\supp(T) \leq \supp(T'') + r$. This is an inequality because adding the edge~$(w, z)$ makes~$z$ non-leaf, thus a support vertex might be removed. Therefore, we have~$\psi(T) \leq \psi(T'') + 2r + 2$. And:
    
                \begin{displaymath}
                \begin{aligned}
                    \gamma(T) &= \gamma(T'') + r \\
                    &\geq \frac{\lvert \mathit{Deg}_{\geq 2}(T'') \rvert + \lvert \supp(T'') \rvert}{4} + r \,\,\, \quad \textnormal{[by induction hypothesis]} \\
                    &= \frac{\lvert \mathit{Deg}_{\geq 2}(T'') \rvert + \lvert \supp(T'') \rvert + 4r}{4} \\
                    &\geq \frac{\lvert \mathit{Deg}_{\geq 2}(T) \rvert + \lvert \supp(T) \rvert}{4}\,,
                \end{aligned}
                \end{displaymath}
                where the last inequality holds because~$4r \geq 2r + 2$ when~$r \geq 1$, which is always true since there is a path from~$w$ to~$x$.     
            \end{proof}
    
            The lower bound in Proposition \ref{lemma:treeLowerBound} is asymptotically tight. Consider the tree $T$ containing a vertex~$v$ with one leaf and~$r$ paths of length 4 ($P_4$) incident on it, which can be also viewed as a star graph with~$r+1$ leaves except~$r$ of them are replaced with~$P_4$. Clearly, we have $\gamma(T)=r+1$. Meanwhile, we have $|\mathit{Deg}_{\geq 2}(T)|=3r+1$, and~$\lvert \supp(T) \rvert = r + 1$. Hence,
            \begin{displaymath}
                \frac{\lvert \mathit{Deg}_{\geq 2}(T) \rvert + \lvert \supp(T) \rvert}{4} =  \frac{4r + 2}{4} \approx r + \frac{1}{2}\,.
            \end{displaymath}
            When~$r \rightarrow \infty$, the ratio of ${r + 1}$ to $r + \frac{1}{2}$ limits to~$1$.
    
            Combining Proposition \ref{lemma:graphUpperBound} and Proposition~\ref{lemma:treeLowerBound}, we have Theorem \ref{theo:dom-num_2approx}. Moreover, combining Theorem \ref{theo:dom-num_2approx} and Theorem \ref{theo:3approxDom-forest}, if~$\lvert \supp(F) \rvert \leq \frac{\lvert \mathit{Deg}_{\geq 2}(F) \rvert}{3}$, we have
    
            \begin{displaymath}
                \frac{\lvert \mathit{Deg}_{\geq 2}(F) \rvert}{3} \leq \gamma(F) \leq \frac{\lvert \mathit{Deg}_{\geq 2}(F) \rvert + \lvert \supp(F) \rvert}{2} \leq \frac{2\lvert \mathit{Deg}_{\geq 2}(F) \rvert}{3}\,,
            \end{displaymath}
            whence Corollary \ref{coro:2approxDS-nohusu} follows.
            \begin{corollary}
            \label{coro:2approxDS-nohusu}
                If~$\lvert \supp \rvert \leq \frac{1}{3} \, \lvert \mathit{Deg}_{\geq 2}\rvert$, then $\frac{1}{3} \, \lvert \mathit{Deg}_{\geq 2} \rvert \leq \gamma \leq \frac{2}{3} \, \lvert \mathit{Deg}_{\geq 2} \rvert \,$.
            \end{corollary}

        \subsubsection{Structural Bounds on the Matching Number $\phi$}

            Similarly, previous works showed that matching number, $\phi$, can be $2$-approximated.

            \begin{theorem}
            \label{theo:2approx-matching-forest}
                \cite{bury2015sublinear,esfandiari2018streaming} For every forest~$F$, $\max \big\{\,c\,,\,\frac{1}{2} (\lvert \mathit{Deg}_{\geq 2} \rvert + c) \ \big\} \leq \phi \leq \lvert \mathit{Deg}_{\geq 2} \rvert + c\,$.
            \end{theorem}
            
            We then show that the approximation ratios above can be further improved to $3/2$ using support vertices.

            \begin{theorem}
            \label{matching_Forest_FullBound}
                For every forest $F$, $\frac{1}{3} \big( \lvert \mathit{Deg}_{\geq 2} \rvert + \lvert \supp \rvert \big) \leq \phi \leq \frac{1}{2} \big( \lvert \mathit{Deg}_{\geq 2} \rvert + \lvert \supp \rvert \big) \,$.
            \end{theorem}

            \begin{proof}
                By Lemma \ref{mat_num_edge_leaf_supp}, we will assume that (unless~$n < 2$) every support vertex is matched via at least one edge in the maximum matching. We first show a upper bound (Proposition~\ref{matching_Tree_UpperBound}) and lower bound (Proposition~\ref{matching_Tree_LowerBound}) on the matching number of trees.  Theorem~\ref{matching_Forest_FullBound} then follows immediately by applying both propositions to each tree component of the forest $F$.
            \end{proof}

            \begin{lemma}
            \label{mat_num_edge_leaf_supp}
                In every connected graph with~$n \geq 3$, there is a maximum matching such that every support vertex is matched by one of its edges between it and its leaves.
            \end{lemma}

            \begin{proof}
                The first half of the lemma (every support vertex is matched) can be proved via contradiction. Assume we have a maximum matching and a support vertex~$s$ is not matched, then we could include an edge between~$s$ and one of its leaves into the matching, a contradiction.

                The latter half can be proved via alternating edges in every given maximum matching without decreasing its size. Given a maximum matching~$M$. If~$M$ satisfies the condition, then we are done. If there is a support vertex~$s$ that is matched through an edge between it and a non-leaf vertex, then by the definition of matching, none of its leaves are matched. Hence, we can remove the matching edge of~$s$ and include one of the edges between~$s$ and its leaves into the matching. This is valid as the leaf is not matched before, and~$s$ is matched only by the newly added edge. The size of maximum matching does not decrease.
            \end{proof}
            We may now prove an upper bound on~$\phi$.

            \begin{proposition}
            \label{matching_Tree_UpperBound}
                For every tree~$T$ with~$n \geq 2$,

                \begin{displaymath}
                    \phi \leq \frac{\lvert \mathit{Deg}_{\geq 2}(T) \rvert + \lvert \supp(T) \rvert}{2}\,.
                \end{displaymath}
            \end{proposition}

            \begin{proof}
                We prove the proposition by induction on $n=|T|$. If~$|T| = 2$, then $T$ is an edge and we have~$\phi(T) = 1, \lvert \mathit{Deg}_{\geq 2}(T) \rvert = 0$ and $\lvert S \rvert = 2$: hence the inequality holds. Henceforth, we assume that~$n > 2$.
                Rearranging the upper bound, we have
                \begin{displaymath}
                    \lvert \supp(T) \rvert + \frac{\lvert \mathit{Deg}_{\geq 2}(T) \rvert - \lvert \supp(T) \rvert}{2} \, .
                \end{displaymath}

                By Lemma \ref{mat_num_edge_leaf_supp}, all support vertices are matched by the edges between them and one of their leaves. Hence, removing all support vertices and their leaves removes~$\lvert S \rvert$ edges in the matching. The remaining graph is a forest~$F$ with~$n(F) = \lvert \mathit{Deg}_{\geq 2}(T) \rvert - \lvert S \rvert$, and all edges in~$F$ can be added into the matching as they do not share a common endpoint with edges between support vertices and leaves (i.e.,~the matching edges that are removed). Hence it remains to prove that for every forest, there is a matching of size at most~$\frac{n(F)}{2}$. Since every forest is bipartite, and in a bipartite graph with total order of~$n$, the matching number is at most~$\frac{n}{2}$. This is because one side has at most~$\frac{n}{2}$ nodes. By the Pigeonhole Principle, if there is a matching of size greater than~$\frac{n}{2}$, at least one node has more than one matching edges incident on it, a contradiction.
            \end{proof}

            The upper bound of Proposition~\ref{matching_Tree_UpperBound} is tight: consider the tree $T$ obtained by taking a star graph~$H$ with~$r$ leaves and attaching a new leaf for each vertex in~$H$ (including the center). It is easy to see that~$\lvert \mathit{Deg}_{\geq 2}(T) \rvert = \lvert \supp(T) \rvert = r + 1$, as all vertices in~$H$ become a support vertex and have more than 1 neighbor. And it is not hard to see that~$\phi = r + 1$, as the maximum matching is the edge set~$E(T) \setminus E(H)$.

            The next proposition gives a lower bound for the matching number:

            \begin{proposition}
            \label{matching_Tree_LowerBound}
                For every tree~$T$ with~$|T| \geq 2$,

                \begin{displaymath}
                    \phi(T) \geq \frac{\lvert \mathit{Deg}_{\geq 2}(T) \rvert + \lvert \supp(T) \rvert}{3}\,.
                \end{displaymath}
            \end{proposition}

            \begin{proof}

                We prove the proposition by induction on $n=|T|$. We consider trees with~$2 \leq n \leq 4$ as the base cases, which can be easily verified by hand.
                There are two cases to consider depending on whether not $T$ contains twins.

                \paragraph*{Case 1: There exist twins in $T$}
                Let $x,y$ be the leaves which are twins and let $s$ be the common vertex that they are adjacent to. Remove either of $x$ or $y$, and let~$T'$ be the tree after the removal. Note that in $T'$,~$s$ is a leaf if and only if~$d(s) = 2$ in $T$, in this case~$n = 3$, which falls into our base case. Hence, we have ~$\lvert \mathit{Deg}_{\geq 2}(T) \rvert=\mathit{Deg}_{\geq 2}(T')$ and~$\lvert \supp(T) \rvert = |\supp(T')|$. Moreover, we claim that $\phi(T)=\phi(T')$: this is because only one of~$s$'s edges can be included into matching, if the removed leaf is on this edge, then we could add another edge between~$s$ and the other one of twin. By the induction hypothesis, this inequality holds.

                \paragraph*{Case 2: There are no twins in $T$} 
                Consider the longest path in~$T$, denote one of its endpoint as~$x$, the parent of~$x$ as~$y$, the parent of~$y$ as~$w$, and the parent of~$w$ as~$z$. By the induction hypothesis,~$T$ has more than 4 nodes and no twins, hence~$z$ is guaranteed to be a non-leaf. This can be seen by contradiction. Assuming~$z$ is leaf, by our choice of longest path, the neighbors of~$w$ and~$y$ can only be leaves. Otherwise, there exists a longer path. Also, since there is no twin,~$w$ and~$y$ both have degree 2, thus there is exactly 4 nodes, a contradiction. By similar argument, it is not hard to see that~$y$ is guaranteed to have degree 2.

                \paragraph*{Case 2.1:~$d(w) > 2$} 
                We remove both~$x$ and~$y$, denote the resulting tree as~$T'$. Since~$d(w) > 2$, removing~$x$ and~$y$ does not make $w$ leaf, thus~$\lvert  \mathit{Deg}_{\geq 2}(T) \rvert = \lvert \mathit{Deg}_{\geq 2}(T') \rvert + 1$. Also,~$w$ is a support vertex in $T'$ if and only if it is a support vertex in~$T$, we have~$\lvert \supp(T) \rvert = \lvert \supp(T') \rvert + 1$. By Lemma \ref{mat_num_edge_leaf_supp}, we know that~$e(x, y)$ is in the maximum matching of~$T$, but not~$e(y, w)$. Hence, the removal of~$x$ and~$y$ does not change the matching status of vertices in~$T'$: this implies that ~$\phi(T) = \phi(T') + 1$. By the induction hypothesis, the inequality holds.

                \paragraph*{Case 2.2:~$d(w) = 2$} 
                We have two cases to consider:

                \paragraph*{Case 2.2.1:~$z \notin \supp(T)$}
                We remove~$x$ and~$y$. Similar to Case 2.1, we have~$\phi(T) = \phi(T') + 1$. However, since~$d_w = 2$ and~$z$ is not a support vertex in~$T$, the removal of~$x$ and~$y$ makes~$w$ a leaf and~$z$ a support vertex. Hence~$\lvert \mathit{Deg}_{\geq 2}(T) \rvert = \lvert \mathit{Deg}_{\geq 2}(T') \rvert + 2$, and~$\lvert \supp(T) \rvert = \lvert \supp(T') \rvert$. By the induction hypothesis, the inequality holds.

                \paragraph*{Case 2.2.2:~$z \in \supp(T)$}
                Consider~$T'$ as the tree after removing~$x$,~$y$, and~$w$ from~$T$. Since~$z$ is a support vertex, according to Lemma \ref{mat_num_edge_leaf_supp},~$z$ must be matched and~$e(x,y)$ is in the maximum matching. Hence, we know that~$e(y, w)$ and~$e(w, z)$ are not in the matching, and the removal of~$x$,~$y$, and~$w$ does not change the matching status of vertices in~$T'$. Therefore,~$\phi(T) = \phi(T') + 1$. Moreover, removing these three vertices does not make~$z$ a leaf or make it not a support vertex in~$T'$, unless~$z$ itself has degree 2. In the former case, we have~$\lvert \mathit{Deg}_{\geq 2}(T) \rvert = \lvert \mathit{Deg}_{\geq 2}(T') \rvert + 2$, and~$\lvert \supp(T) \rvert = \lvert \supp(T') \rvert + 1$. And by the induction hypothesis, the inequality holds. In the later case, we have a path on $5$ veryices since~$x$ is a leaf,~$y$,~$w$, and~$z$ have degree 2, and~$z$ is also a support vertex. It can be easily verified by hand that the inequality holds for~$P5$.
            \end{proof}

            The lower bound is asymptotically tight. Consider a graph~$G$ constructed as follows. Start with a star graph,~$H$, with~$r$ leaves, and replace each leaf with a path of length~$3$.~$G$ has~$2 r + 1$ non-leaf vertices, including the star center and two of the vertices on each of the paths. Also~$\lvert \supp(G) \rvert = r$ as only the middle vertices on the paths are support vertices. Lastly, a maximum matching can be found by adding all edges between support vertices and leaves, and one edge from the edges incident on star center, thus~$\phi(G) = r + 1$.

            Theorem~\ref{matching_Forest_FullBound} follows if we combine Proposition~\ref{matching_Tree_UpperBound} and Proposition~\ref{matching_Tree_LowerBound}. Moreover, combining Theorem~\ref{theo:2approx-matching-forest} and Theorem~\ref{matching_Forest_FullBound}, if~$\lvert \supp(F) \rvert \leq \frac{\lvert \mathit{Deg}_{\geq 2}(F) \rvert}{3}$, we have:

            \begin{displaymath}
                \phi(F) \leq \frac{\lvert \mathit{Deg}_{\geq 2}(F) \rvert + \lvert \supp(F) \rvert}{2} \leq \frac{3 (\lvert \mathit{Deg}_{\geq 2}(F) \rvert + c(F))}{4}\,.
            \end{displaymath}
            Since, by Theorem \ref{theo:2approx-matching-forest}, $\phi(F) \geq \max \{ c(F), \frac{\lvert \mathit{Deg}_{\geq 2}(F) \rvert + c(F)}{2} \}$, Corollary \ref{coro:4/3approxMM-noS} follows.

            \begin{corollary}
            \label{coro:4/3approxMM-noS}
                If~$\lvert \supp \rvert \leq \frac{1}{2} \, \lvert \mathit{Deg}_{\geq 2} \rvert$, then $\max \big\{ c, \frac{\lvert \mathit{Deg}_{\geq 2} \rvert + c}{2} \big\} \leq \phi \leq \frac{3}{4}\cdot \big( \lvert \mathit{Deg}_{\geq 2} \rvert + c \big) \,$.
            \end{corollary}

    \subsection{One-pass and Two-pass Sublinear Streaming Algorithms for Estimating the Independence Number, Domination number and Matching number}
    \label{sec:streaming-algorithms}

    In Section~\ref{subsec:structural-bounds} we obtained (approximate) upper and lower bounds for $\beta, \gamma$ and $\phi$ in terms of other forest parameters such as number of leaves, number of non-leaves, number of support vertices, etc. Now we estimate these latter quantities in sublinear space to obtain one-pass and two-pass streaming algorithms which estimate $\beta, \gamma$ and $\phi$. 

    \subsubsection{One-pass and Two-pass Streaming Algorithms for Estimating $\beta$ }

        We show how to estimate the structural graph quantities within factor~$(1 \pm \eps)$, in small space.
        Our algorithms rely on the concept of the degree vector~$\degv(F)$:
        the~$i^{\text{th}}$ coordinate holds the degree of vertex~$i$.
        Storing this vector would require~$\Theta(n \log n)$ bits, exceeding our working-space bound.
        Hence we rely on sparse-recovery techniques to estimate quantities that are functions of~$\degv(F)$, supporting fast stream updates.
        
        \paragraph*{Estimating Non-leaf Vertices}
        Though $\beta$ estimation does not use~$\lvert \mathit{Deg}_{\geq 2} \rvert$, it helps to estimate~$\gamma$ and~$\phi$. 
        Consider the vector $\degv'(F) = \degv(F) + \{-1\}^n$.
        Since we assume~$F$ has no isolated vertices, the number of non-zero entries in~$\degv'(F)$ is the number of vertices of degree at least~$2$, viz.  
        $ \Vert \degv'(F) \Vert_0 = \lvert \mathit{Deg}_{\geq 2} \rvert$.
        Applying Theorem~\ref{theo:L0}, we~$(1 \pm \eps)$-approximate~$\lvert \mathit{Deg}_{\geq 2} \rvert$ in~$\log^{\mathcal{O}(1)} n$ space with constant success probability.
        Since no existing~$L_0$ sketch supports initialization with~$\{-1\}^n$, degree decrements are post-processed.
        Running $\mathcal{O}(\log \delta^{-1})$ independent instances and returning the median, the failure probability is less than~$\delta$.
            
        \begin{lemma}
        \label{lemma:H2}
            Given $\eps, \delta\in (0,1)$,~$\lvert \mathit{Deg}_{\geq 2} \rvert$ is~$(1 \pm \eps)$-approximated in a turnstile stream with probability~$(1 - \delta)$ and in $\mathcal{O}(\log^{\mathcal{O}(1)} n \cdot \log \delta^{-1})$ space.
        \end{lemma}

        \paragraph*{Estimating Leaves}
        We can estimate $\lvert \mathit{Deg}_{1} \rvert$ with $n- \lvert \mathit{Deg}_{\geq 2} \rvert$.
        But if~$\lvert \mathit{Deg}_{1} \rvert$ is in~$o(\lvert \mathit{Deg}_{\geq 2} \rvert)$, the relative error can be very large.
        We turn elsewhere.

        \begin{lemma}
        \label{lemma:forestLeaves}
            For every forest, $\lvert \mathit{Deg}_{1} \rvert = 2 c + \sum^{\Delta}_{i = 3} (i-2)\cdot\lvert \mathit{Deg}_{i} \rvert\,$.
        \end{lemma}

        \begin{proof}
            We prove this lemma by showing the following proposition:
            \begin{proposition}
                For each tree $T$ (with at least two vertices) we have $\lvert \mathit{Deg}_{1}(T) \rvert = 2 + \sum^{\Delta}_{i = 3} (i-2)\cdot\lvert \mathit{Deg}_{i}(T) \rvert$.
                \label{clm:leafy}
            \end{proposition}
            
            \begin{proof}
                \begin{displaymath}
                \begin{aligned}
                    &\sum^{\Delta}_{i = 3} (i-2)\cdot\lvert \mathit{Deg}_{i}(T) \rvert = \sum^{\Delta}_{i = 2} (i-2)\cdot\lvert \mathit{Deg}_{i}(T) \rvert \\
                    &= \sum^{\Delta}_{i = 2} i\cdot\lvert \mathit{Deg}_{i}(T) \rvert - 2\cdot \sum^{\Delta}_{i = 2} \lvert \mathit{Deg}_{i}(T) \rvert\\
                    &= \Big( (2|T|-2) - \mathit{Deg}_{1}(T) \Big) - 2\cdot \Big(|T|-\mathit{Deg}_{1}(T)\Big) \\
                    &= \mathit{Deg}_{1}(T) -2\,.
                \end{aligned}
                \end{displaymath}
                In the penultimate line, we have used the fact that $\sum_{v\in T} d(v)=2|T|-2$.
            \end{proof}
            
            The proof of Lemma~\ref{lemma:forestLeaves} then follows from Proposition~\ref{clm:leafy} by summing up over all tree components of $F$.
        \end{proof}
            
        Post-processing~$\degv$ with~$\{-2\}^n$, to obtain $\degv''$, we have
        $\Vert\degv''(F)\Vert_1 = \lvert \mathit{Deg}_{1} \rvert + \sum_{i = 3}^{\Delta} (i - 2) \cdot \lvert \mathit{Deg}_{i} \rvert \,$.
        Folding in Lemma~\ref{lemma:forestLeaves}, we have~$\lvert \mathit{Deg}_{1} \rvert = \Vert \degv''(F) \Vert_1 /{2} + c$,
        so via Theorem~\ref{theo:Lp} we have
        \begin{lemma}
        \label{lemma:Deg1}
            Given~$\eps, \delta \in (0, 1)$,~$\lvert \textit{Deg}_{1}\rvert$ is $(1 \pm \eps)$-approximated in a turnstile stream with probability $(1-\delta)$ in~$\mathcal{O}(\log^{\mathcal{O}(1)} n \cdot \log \delta^{-1})$ space.
        \end{lemma}

        \noindent
        Combining Lemma~\ref{lemma:Deg1} with Theorem~\ref{theo:3/2approxIND-forest}, we conclude:        
        \begin{theorem}
        \label{theo:ind-3/2approx}
            Forest independence number~$\beta$
            is $3/2 \cdot (1 \pm \eps)$-approximated with probability~$(1 - \delta)$ on turnstile streams using~$\mathcal{O}(\log^{\mathcal{O}(1)} n \cdot \log \delta^{-1})$ space.
        \end{theorem}

        \paragraph*{Estimating Support Vertices}
        Degree-counting, above, does not suffice to estimate the number of support vertices.
        We first show, in Subroutine~\ref{alg:SLarge}, given a turnstile streamed forest, how to output an~$(1 \pm \eps)$-estimate of~$\lvert \supp \rvert$, when~$\lvert \supp \rvert$ is at least a threshold~$K_1$.
        When~$\lvert \supp(F) \rvert$ is small, we can  approximate~$\beta$ via either of the following methods:
        
        \begin{itemize}
            \item \updates{If few non-leaves, find~$|\mathit{Deg}_{\geq 2} \rvert$ and~$\lvert \supp(F) \rvert$ exactly (Subroutine~\ref{alg:dom-H2-Small}).}
            \item \updates{If many non-leaves, Corollary~\ref{coro:4/3approxIND-noS} enables excluding~$\lvert \supp(F) \rvert$}.
        \end{itemize}
        
        Similar to Cormode et al.~\cite{cormode2017sparse} and to Jayaram and Woodruff~\cite{DBLP:conf/pods/JayaramW18}, we assume the number of deletions is~$\mathcal{O}(n)$, preserving the analyses of Subroutines~\ref{alg:SLarge} and~\ref{alg:dom-H2-Small}.

        \begin{algorithm}
            \floatname{algorithm}{Subroutine}

            \caption{Estimating~$\lvert \supp \rvert$}     
            \label{alg:SLarge}
            \begin{algorithmic}[1]

            \State{\textbf{Input}: Size threshold~$K_1$, constant~$c_1$, and error rate~$\eps_1 \in (0, 1)$}

            \State{\textbf{Initialization}: $I \gets \frac{c_1 n}{\eps_1^2 K_1}$ vertices sampled uniformly at random from~$[n]$}
            \For{$v \in I$}
            \State{$l(v) \gets \{\}$}
            \EndFor
            \State{$m,t \gets 0$}
            \State{\textbf{First Pass}:}
            \ForAll{$(e = (u, v),i)$ in the stream, with $i \in \pm 1$}
                \State{$m \gets m + i$}
                \If{$u \in I$}
                    \State{Toggle~$v$'s presence in~$l(u)$}
                    \State{$t \gets t + i$}
                \EndIf
                \State{Perform the same operation on~$v$}
                \State{Abort if~$t \geq \frac{2 m}{n} \frac{c_1 n}{\eps_1^2 K_1} e^{\frac{c_1}{3}}$}
            \EndFor

            \State{\textbf{Second Pass}: }
            \State{Count the degree of every vertex in~$I  \cup \left(\bigcup_{w \in I}\, l(w) \right)$}
            \State{$C \gets \{u \mid u \in I \text{ and there is some leaf $v \in l(u)$\}}$}
            \State{\textbf{return}~$\widehat{\lvert \supp(F) \rvert} \gets \lvert C \rvert \times \frac{\eps_1^2 K_1}{c_1}$}

            \end{algorithmic}
        \end{algorithm}

        \begin{lemma}
        \label{lemma:SLarge}
            Given threshold~$K_1$, constant~$c_1$, and error rate~$\eps_1 \in (0, 1)$, Subroutine~\ref{alg:SLarge} is a turnstile-stream algorithm that uses~$\widetilde{\mathcal{O}}({n}/(\eps_1^2 K_1))$ space and:
            when~$\lvert \supp(F) \rvert \geq K_1$,~$(1 \pm \eps_1)$-approximates~$\lvert \supp(F) \rvert$ with probability~$\geq 1 - 3 e^{-\frac{c_1}{3}}$;
            when~$\lvert \supp(F) \rvert < K_1$,  returns an estimate at most~$2 K_1$ with probability~$\geq 1 - 2 e^{-\frac{c_1}{3}}$.
        \end{lemma}

        \begin{proof}
            We claim that, conditional on the algorithm does not abort at line 13, the candidate set,~$C$, obtained at line 16 contains only vertices from~$\supp(F)$, and if a vertex in~$\supp(F)$ is sampled, it is kept in~$C$. This is not hard to see, by performing a second pass, we obtain the exact degree of all sampled vertices and their neighbors. So that a vertex is included in~$C$ if there is at least one leaf adjacent to it, which is exactly the definition of~$\supp(F)$.

            It remains to bound the output of this algorithm, conditional on the algorithm does not abort. For all~$i \in V(G)$, let the binary random variable~$X_i$ indicate whether a vertex~$v_i$ is a support vertex and is sampled in~$I$. Let~$X = \sum_{i \in G} X_i$. Clearly, since vertices are sampled uniformly with probability~$p = \frac{c_1}{\eps_1^2 K_1}$, we have
            \begin{displaymath}
                \begin{aligned}
                E[X] &= \sum_{i \in G} E[X_i] = \sum_{i \in G} \Pr[v_i \in (\supp(F) \cap I)] \\
                &= \sum_{i \in \supp(F)} \Pr[v_i \in I] = \frac{c_1 \, \lvert \supp(F) \rvert}{\eps_1^2 K_1}\,.
                \end{aligned}
            \end{displaymath}
            Thus~$E[\lvert \widehat{S} \rvert] = \frac{\eps_1^2 K_1}{c_1} \times E[X] = \lvert \supp(F) \rvert$. Moreover, when~$\lvert \supp(F) \rvert \geq K_1$, we have
            \begin{displaymath}
            \begin{aligned}
                \Pr \big[ \big\lvert \lvert \widehat{S} \rvert - \lvert \supp(F) \rvert \big\rvert \geq \eps_1 \lvert \supp(F) \rvert \big] &= \Pr \big[ \lvert X - E[X] \rvert \geq \eps_1 E[X] \big] \\
                &\leq 2 e^{- \frac{\eps_1^2 E[X]}{3}} \\
                &\leq 2 e^{- c_1/3} \, ,
            \end{aligned}
            \end{displaymath}
            where we apply the negatively correlated Chernoff bound (Theorem \ref{theorem:neg-chernoff}) in the second-last inequality, since the~$X_i$'s are negatively correlated. And the last inequality holds as~$\lvert \supp(F)\rvert \geq K_1$.

            Next, we show that when~$\lvert \supp(F) \rvert < K_1$, conditional on the algorithm does not abort, the probability that~$\lvert \widehat{S} \rvert$ is larger than~$2 K_1$ is small. Again, by the negatively correlated Chernoff bound, we have
            \begin{displaymath}
            \begin{aligned}
                \Pr \big[ \lvert \widehat{S} \rvert > 2 K_1 \big]
                &= \Pr \big[ X > \frac{2 K_1 E[X]}{\lvert \supp(F) \rvert} \big] \\
                &< \Pr \big[ X > (1 + \frac{K_1}{\lvert \supp(F) \rvert}) E[X] \big] \\
                &\leq e^{- \frac{K_1^2 E[X]}{3 \lvert \supp(F) \rvert^2}} \\
                &\leq e^{- \frac{c_1 K_1}{3 \lvert \supp(F) \rvert \eps^2}} \\
                &\leq e^{- \frac{c_1}{3 \eps^2}} \, ,
            \end{aligned}
            \end{displaymath}
            where the first and the last inequalities hold because~$\lvert \supp(F) \rvert < K_1$.

            Lastly, we bound the space usage of the algorithm and the probability that the algorithm aborts at line 13. By the design of the algorithm, at most~$\frac{2 m}{n} \frac{c_1 n}{\eps_1^2 K_1} e^{\frac{c_1}{3}}~$ vertices can be stored from the first pass, otherwise it aborts at line 13. By our assumption, the number of total edge insertions is bounded by~$\mathcal{O}(n)$. Thus, we have~$\frac{2m}{n} \in \mathcal{O}(1)$. Therefore, since each vertex requires~$\mathcal{O}(\log n)$ bits to store its identifier, the space usage of the first pass is~$\widetilde{\mathcal{O}}(\frac{n}{\eps_1^2 K_1})$. Moreover, in the second pass, each degree counter takes an additional~$\mathcal{O}(\log n)$ space to store. Hence, the total space usage is still~$\widetilde{\mathcal{O}}(\frac{n}{\eps_1^2 K_1})$.

            In this algorithm, we sampled~$\frac{c_1 n}{\eps_1^2 K_1}$ vertices in advance. And the average degree of the graph at any point of the stream can be calculated as~$\frac{2m}{n}$. Hence, if we pick a vertex uniformly at random, the expected number of its neighbors is bounded by~$\frac{2m}{n}$. Let~$Y_i$ be the number of neighbors of vertex~$i$,~$E[Y_i] = \frac{2m}{n}$. Let~$Y = \sum_{v \in V(G)} Y_v$ We have

            \begin{displaymath}
                E[Y] = \sum_{v \in V(G)} E[Y_v] = \frac{2m}{n} \frac{c_1 n}{\eps_1^2 K_1}\,.
            \end{displaymath}

            The algorithm aborts if more than~$\frac{2 m}{n} \frac{c_1 n}{\eps_1^2 K_1} e^{\frac{c_1}{3}}$ vertices are retained. Since we are sampling each vertex uniformly at random,~$Y_v$'s are independent of each other. By Markov inequality, the probability that the actual size is~$e^{\frac{c_1}{3}}$ times larger than the expected size is at most~$e^{-\frac{c_1}{3}}$.

            Applying a union bound with the abort probability and the concentration probabilities above, we can bound the fail probability of our algorithm as follows. When~$\lvert \supp(F) \rvert~\geq~K_1$, Algorithm \ref{alg:SLarge} outputs an~$(1 \pm \eps)$-estimate of~$\lvert \supp(F) \rvert$ and fails with probability at most~$2 e^{-\frac{c_1}{3}} + e^{-\frac{c_1}{3}} = 3 e^{-\frac{c_1}{3}}$. Similarly, when~$\lvert \supp(F) \rvert < K_1$, the estimate returned by Algorithm \ref{alg:SLarge} is larger than~$2 K_1$ with probability at most~$e^{-\frac{c_1}{3 \eps_1^2}} + e^{-\frac{c_1}{3}} < 2 e^{-\frac{c_1}{3}}$ since~$\eps_1 < 1$.
        \end{proof}

        Sparse recovery on~$\degv'$ recovers all vertices in~$\mathit{Deg}_{\geq 2}$, with probability~$1 - \delta$, but no leaves.
        In a second pass, we verify whether each recovered vertex is a support vertex.
        For a length-two path (a~$P_2$, i.e., an isolated edge), neither endpoint is recovered.
        Except in the case of a~$P_2$ path, if some neighbour of vertex~$u$ is not recovered, then~$u$ is a support vertex.
        Since no leaf is recovered,
        Subroutine~\ref{alg:dom-H2-Small} shows how to verify~$\mathit{Deg}_{\geq 2}$.

        \begin{algorithm}
            \floatname{algorithm}{Subroutine}

            \caption{Estimating~$\lvert \supp(F) \rvert$ when~$\lvert \mathit{Deg}_{\geq 2}(F) \rvert \leq K_2$}
            \label{alg:dom-H2-Small}
            \begin{algorithmic}[1]

            \State{\textbf{Input}: A size threshold~$K_2$, a large constant~$c_2$}

            \State{\textbf{Initialization}: Sparse recovery sketch~$\sketch$ in \cite{cormode2014unifying} with~$k = K_2$ and~$\delta = 1/c_2$}

            \State{\textbf{First Pass}:}
            \ForAll{$(e = (u, v),i)$ in the stream, $i \in \pm 1$}
                \State{Update~$\sketch$ with~$(e,i)$}
            \EndFor

            \State{Decrement all coordinates in $\sketch$ by~$1$}
            \State{Denote the result of sparse recovery from~$\sketch$ by~$\recovered$}
            \State{$p,m \gets 0$}
            \For{$v \in \recovered$}
            \State{$d_v,l_v \gets 0$} \Comment{Degree and leaf counters}
            \EndFor

            \State{\textbf{Second Pass}:}
            \ForAll{$(e = (u, v),i)$ in the stream, $i \in \pm 1$}
                \State{$m \gets m+i$;
                $d_u \gets d_u + i$; 
                $d_v \gets d_v + i$}
                \If{$u \notin \recovered$ and~$v \notin \recovered$}
                    \State{$p \gets p+i$}
                \EndIf
                \If{$u \in \recovered$ and~$v \notin \recovered$}
                    \State{$l_u \gets l_u +i$}
                \EndIf
                \If{$v \in \recovered$ and~$u \notin \recovered$}
                    \State{$l_v \gets l_v +i$}
                \EndIf
            \EndFor

            \If{$2 m \neq n - \lvert \recovered \rvert + \sum_{v \in \recovered} d_v~$}
                \State{\textbf{Return} \textit{FAIL}}
            \Else{}
                \State{$\widehat{\lvert \supp(F) \rvert} \gets \lvert \{u \mid u \in \recovered \text{ and } l(u) = 1 \} \rvert + 2 p$}
                \State{$\widehat{\lvert \mathit{Deg}_{\geq 2}(F) \rvert} \gets \lvert \recovered \rvert$}
                \State{\textbf{Return}~$\widehat{\lvert \supp(F) \rvert}$ and~$\lvert \widehat{\mathit{Deg}_{\geq 2}} \rvert$}
            \EndIf

            \end{algorithmic}
        \end{algorithm}

        \begin{lemma}
        \label{lemma:dom-H2-Small}
            Given threshold~$K_2$, constant~$c_2$, Subroutine~\ref{alg:dom-H2-Small} in two passes returns both~$\lvert \supp(F) \rvert$ and~$\lvert \mathit{Deg}_{\geq 2}(F) \rvert$ in~$\widetilde{\mathcal{O}}(K_2)$ space. When~$\lvert \mathit{Deg}_{\geq 2}(F) \rvert \leq K_2$, the failure probability is at most~$1/c_2$.
        \end{lemma}

        \begin{proof}
            Note that in the first pass, only vertices in~$\mathit{Deg}_{\geq 2}(F)$ are sampled. This is because line 6 decreases the frequency of all vertices by one, so leaves has 0 frequency and by the definition of~$k$-sparse recovery, they are not recovered by the sparse recovery data structure. Moreover, by Theorem \ref{theo:k-sparse-recovery}, when~$\lvert \mathit{Deg}_{\geq 2}(F) \rvert \leq K_2$ (i.e., when it is small), the sparse recovery fails with probability no more than~$1/c_2$.

            Since the sparse recovery data structure might output a false positive result. That is, a~$k$-sparse vector when the original degree vector is not~$k$-sparse. We need to perform a sanity check at line 20 to ensure that the returned result contains all~$\mathit{Deg}_{\geq 2}(F)$ vertices, rather than part of them. For this sanity check to work, we prove that the equality holds if and only if~$H = \mathit{Deg}_{\geq 2}(F)$.

            To begin with, we prove that the equality holds if~$H = \mathit{Deg}_{\geq 2}(F)$. This is trivial as for every graph~$G$, we have the following relationship between its edge count,~$m$, and its total degree count:
            \begin{displaymath}
                2 m = \sum_{v \in V(G)} d(v) = n - \lvert \mathit{Deg}_{\geq 2}(G) \rvert + \sum_{v \in \mathit{Deg}_{\geq 2}(G)} d(v)\,.
            \end{displaymath}

            Next, we prove that the equality does not hold if~$H \neq \mathit{Deg}_{\geq 2}(F)$. Suppose~$H \neq \mathit{Deg}_{\geq 2}(F)$, note that by our argument above, it is impossible to have $x$ such that $x \in H$ but $x \notin \mathit{Deg}_{\geq 2}(F)$. Hence, we only consider the case where~$\exists x \in \mathit{Deg}_{\geq 2}(F) \text{ s.t. } x \notin H$. We have

            \begin{displaymath}
            \begin{aligned}
                &2m - (n - \lvert H \rvert + \sum_{v \in H} d(v))\\
                &= n - \lvert \mathit{Deg}_{\geq 2}(F) \rvert + \sum_{v \in \mathit{Deg}_{\geq 2}(F)} d(v) - (n - \lvert H \rvert + \sum_{v \in H} d(v)) \\
                &= \sum_{v \in \mathit{Deg}_{\geq 2}(F)} d(v) - \lvert \mathit{Deg}_{\geq 2}(F) \rvert - \sum_{v \in H} d(v) + \lvert H \rvert \\
                &= \sum_{v \in \mathit{Deg}_{\geq 2}(F) \setminus H} d(v) - \lvert \mathit{Deg}_{\geq 2}(F) \rvert - \sum_{v \in H \setminus \mathit{Deg}_{\geq 2}(F)} d(v) + \lvert H \rvert \\
                &= \sum_{v \in \mathit{Deg}_{\geq 2}(F) \setminus H} d(v) - \lvert \mathit{Deg}_{\geq 2}(F) \rvert - \lvert H \setminus \mathit{Deg}_{\geq 2}(F) \rvert + \lvert H \rvert \\
                &= \sum_{v \in \mathit{Deg}_{\geq 2}(F) \setminus H} d(v) - \lvert \mathit{Deg}_{\geq 2}(F)
                \setminus H \rvert \\
                &> 0\,.
            \end{aligned}
            \end{displaymath}
 The second-last equality holds because by our definition of~$\mathit{Deg}_{\geq 2}(F)$, all non-leaf vertices are in~$\mathit{Deg}_{\geq 2}(F)$, hence vertices in~$H \setminus \mathit{Deg}_{\geq 2}(F)$ are leaves and have degree of 1. And the last inequality holds because vertices in~$\mathit{Deg}_{\geq 2}(F)$ have degree at least 2, and~$\mathit{Deg}_{\geq 2}(F) \setminus H \neq \emptyset$. Hence, if some~$\mathit{Deg}_{\geq 2}(F)$ vertices are not sampled, i.e.,~$\mathit{Deg}_{\geq 2}(F) \setminus H \neq \emptyset$, then~$2m \neq n - \lvert H \rvert + \sum_{v \in H} d(v)$.

            It remains to prove that, conditional on all~$\mathit{Deg}_{\geq 2}(F)$ vertices being sampled, both~$\lvert \widehat{S} \rvert$ and~$\lvert \widehat{\mathit{Deg}_{\geq 2}} \rvert$ are correct estimates of~$\lvert \supp(F) \rvert$ and~$\lvert \mathit{Deg}_{\geq 2}(F) \rvert$ respectively. If all~$\mathit{Deg}_{\geq 2}(F)$ vertices are sampled successfully, we can infer whether a vertex is a leaf or not by testing whether it is in~$H$ or not. As indicated by line 16~--~19, by counting the number of leaves,~$l_v$, for each~$v \in H$, we can identify all the support vertices in~$H$. But not all support vertices have degree greater than 1, by definition we can have support vertices of degree 1 in paths of length 2 ($P_2$), such that each~$P_2$ has two support vertices. Thus, we need to also count the number of~$P_2$,~$p$. This is trivial as if an edge arrives with none of its endpoints in~$\mathit{Deg}_{\geq 2}(F)$, this edge must be an instance of~$P_2$.

            Lastly, the space used by Algorithm \ref{alg:dom-H2-Small} is~$\widetilde{\mathcal{O}}(K_3)$. In the first pass, the sparse recovery structure takes~$\widetilde{\mathcal{O}}(K_3)$ space to store. And in the second pass, for each vertex in~$\supp(F)$, we introduce two counters that can be both stored in~$\mathcal{O}(\log n)$ bits. Therefore the total space usage of Algorithm \ref{alg:dom-H2-Small} is~$\widetilde{\mathcal{O}}(K_3)$.
        \end{proof}

        \paragraph*{Finalizing the Algorithms:}
        We run Lemma~\ref{lemma:Deg1}'s algorithm, Subroutine~\ref{alg:SLarge}, and Subroutine~\ref{alg:dom-H2-Small} concurrently, with $K_1 = K_2 = \mathcal{O}(\sqrt{n})$,~$c_1 = \mathcal{O}(\ln \delta^{-1})$, and $c_2 = \mathcal{O}(\delta^{-1})$. Returning the minimum of~$(3 (n + \lvert \textit{Deg}_{1}(F) \rvert))/{8}$ and~$(n + \lvert \textit{Deg}_{1}(F) \rvert - \lvert \supp(F) \rvert)/{2}$, we $4/3\cdot (1 \pm \eps)$-approximate~$\beta$ with probability~$(1 - \delta)$ in~$\widetilde{\mathcal{O}}(\sqrt{n})$ space.

        \begin{algorithm}
        \caption{Estimating~$\beta(F)$ in forest~$F$}
             \label{alg:ind-num-main}
        \begin{algorithmic}[1]

            \State{\textbf{Input}: error rate~$\eps$, fail rate~$\delta$}

            \State{\textbf{Initialization}: $K_1 \gets \sqrt{n}$,~$K_2 \gets 8 \sqrt{n}$,~$c_1 \gets 3 \ln (6 \delta^{-1})$,~$c_2 \gets \delta^{-1}$,~$\eps_1 \gets \eps / 2$}

            \State{Run the following three subroutines concurrently:}
            \State{$\lvert \widehat{\textit{Deg}_1} \rvert \gets \text{ Lemma~\ref{lemma:Deg1} algorithm with~$\eps / 2$ and~$\delta / 2$}$}
            \State{$\widehat{\lvert \supp(F) \rvert} \gets \text{ Subroutine~\ref{alg:SLarge} with~$K_1$,~$c_1$, and~$\eps_1$}$}
            \State{$(\widehat{\lvert \supp(F) \rvert}',\widehat{\lvert \mathit{Deg}_{\geq 2} \rvert}') \gets
            \text{ Subroutine~\ref{alg:dom-H2-Small} with~$K_2$ and~$c_2$}$}

            \If{Subroutine \ref{alg:dom-H2-Small} returns \textit{FAIL}}
                \State{\textbf{Return}~$\widehat{\beta} \gets \min \{ \frac{3 (n + \lvert \widehat{\textit{Deg}_1} \rvert)}{8}, \frac{n + \lvert \widehat{\textit{Deg}_1} \rvert - \lvert \widehat{\supp(F)} \rvert}{2} \}$}
            \Else{}
                \State{$\lvert \widehat{\textit{Deg}_1} \rvert' \gets n - \lvert \widehat{\mathit{Deg}_{\geq 2}} \rvert'$}
                \State{\textbf{Return}~$\widehat{\beta} \gets \min \{ \frac{3 (n + \lvert \widehat{\textit{Deg}_1} \rvert')}{8}, \frac{n + \lvert \widehat{\textit{Deg}_1} \rvert' - \lvert \widehat{\supp(F)} \rvert'}{2} \}$}
            \EndIf
            \end{algorithmic}
        \end{algorithm}
            
        \begin{theorem}
        \label{theo:ind-4/3approx}
            For every $\eps, \delta\in (0,1)$, Algorithm~\ref{alg:ind-num-main} estimates the independence number~$\beta(F)$ in a turnstile forest stream within factor~$4/3\cdot (1 \pm \eps)$ with probability~$(1 - \delta)$ in~$\widetilde{\mathcal{O}}(\sqrt{n})$ space and two passes.
        \end{theorem}

        \begin{proof}

            By Theorem \ref{theo:4/3approxIND} and Corollary \ref{coro:4/3approxIND-noS}, the min between~$\frac{3 (n + \lvert \textit{Deg}_{1}(F) \rvert)}{8}$ and $\frac{n + \lvert \textit{Deg}_{1}(F) \rvert - \lvert \supp(F) \rvert}{2}$ is guaranteed to be a~$4/3$ approximation of the independence number~$\beta(F)$. Hence, it suffices to show that we could either approximate both terms well, or approximate the smaller term well.

            When~$\lvert \mathit{Deg}_{\geq 2}(F) \rvert \leq K_2 = 8 \sqrt{n}$, the algorithm returns at line 11. By Lemma \ref{lemma:dom-H2-Small}, Subroutine~\ref{alg:dom-H2-Small} gives exact values of~$\lvert \mathit{Deg}_{\geq 2}(F) \rvert$ and~$\lvert \supp(F) \rvert$ with probability~$1- 1/c_2 = 1 - \delta$. Also, since~$n = \lvert \mathit{Deg}_{\geq 2}(F) \rvert + \lvert \textit{Deg}_{1}(F) \rvert$ and~$n$ is known, we can obtain an exact value of~$\lvert \textit{Deg}_{1}(F) \rvert$. Hence, no matter which estimate is returned, it is a~$4/3$ approximation of~$\beta(F)$.

            When~$\lvert \mathit{Deg}_{\geq 2}(F) \rvert > 8 \sqrt{n}$, there are two cases to consider. If $\lvert \supp(F) \rvert \geq \sqrt{n}$, we may return the estimate of~$\frac{3 (n + \lvert \textit{Deg}_{1}(F) \rvert)}{8}$ or the estimate of~$\frac{n + \lvert \textit{Deg}_{1}(F) \rvert - \lvert \supp(F) \rvert}{2}$ (line 8). On the one hand, if the estimate of~$\frac{3 (n + \lvert \textit{Deg}_{1}(F) \rvert)}{8}$ is returned, by Lemma \ref{lemma:Deg1} we know that the algorithm on line 4 outputs a~$(1 \pm \frac{\eps}{2})$ estimate of~$\lvert \textit{Deg}_{1}(F) \rvert$ with probability~$1 - \delta / 2$. Thus our result is a~$4/3 (1 \pm \eps/2)$ approximation of~$\beta$.

            On the other hand, if the estimate of~$\frac{n + \lvert \textit{Deg}_{1}(F) \rvert - \lvert \supp(F) \rvert}{2}$ is returned, by Lemma \ref{lemma:SLarge}, Subroutine~\ref{alg:SLarge} gives a ($1 \pm \frac{\eps}{2}$) estimate of~$\lvert \supp(F) \rvert$ with probability
            \begin{displaymath}
                1 - 3 e^{-c_1 / 3} = 1 - 3 e^{- \ln (6 \delta^{-1})} = 1 - \frac{\delta}{2},..
            \end{displaymath}

            Applying a union bound over successfully returning an estimate of~$\lvert \supp(F) \rvert$ and an estimate of~$\lvert \textit{Deg}_{1}(F) \rvert$, the probability of failure is at most~$\delta$. And on the upper-bound side, the approximation ratio is at most

            \begin{displaymath}
            \begin{aligned}
                & \frac{1}{2} \, (n + (1 \pm \frac{\eps}{2}) \lvert \textit{Deg}_{1}(F) \rvert - (1 \pm \frac{\eps}{2}) \lvert \supp(F) \rvert) \\
                &\leq \frac{1}{2} \, (n + \lvert \textit{Deg}_{1}(F) \rvert - \lvert \supp(F) \rvert + \frac{\eps}{2} \, (\lvert \textit{Deg}_1 \rvert + \lvert S \rvert)) \\
                &\leq (1 + \eps) \, \frac{1}{2} \, (n + \lvert \textit{Deg}_{1}(F) \rvert - \lvert \supp(F) \rvert)\,.
            \end{aligned}
            \end{displaymath}
The last inequality holds because~$\lvert \supp(F) \rvert \leq \lvert \textit{Deg}_{1}(F) \rvert \leq n$, thus~$\lvert \textit{Deg}_{1}(F) \rvert + \lvert \supp(F) \rvert$ is no more than~$2 (n + \lvert \textit{Deg}_{1}(F) \rvert - \lvert \supp(F) \rvert)$. Similarly, we prove the approximation ratio on the lower bound as
            \begin{displaymath}
            \begin{aligned}
                &\frac{1}{2} \, (n + (1 \pm \frac{\eps}{2}) \lvert \textit{Deg}_{1}(F) \rvert - (1 \pm \frac{\eps}{2}) \lvert \supp(F) \rvert)\\
                &\geq (1 - \eps) \, \frac{1}{2} \, (n + \lvert \textit{Deg}_{1}(F) \rvert - \lvert \supp(F) \rvert)\,.
            \end{aligned}
            \end{displaymath}
            Hence our algorithm returns a~$4/3\cdot (1 \pm \eps)$ approximation of~$\beta(F)$.

            Lastly, if~$\lvert \supp(F) \rvert < \sqrt{n}$, then by Lemma \ref{lemma:SLarge},~$\lvert \widehat{S} \rvert \leq 2 \sqrt{n}$ with probability~$1 - 2 e^{- \frac{c_1}{3 \eps^2}}$. By Theorem \ref{lemma:H2}, we know that
            \begin{displaymath}
                \lvert \widehat{\mathit{Deg}_{\geq 2}} \rvert \geq (1 - \eps)\cdot \lvert \mathit{Deg}_{\geq 2}(F) \rvert > \frac{\lvert \mathit{Deg}_{\geq 2}(F) \rvert}{2}
            \end{displaymath}
            holds with probability~$1 - \delta / 2$ (if~$\eps < 1/2$). Since~$\lvert \mathit{Deg}_{\geq 2}(F) \rvert \geq 8 \sqrt{n}$, by applying a union bound, we can claim that the probability of~$2 \lvert \widehat{S} \rvert > \lvert \widehat{\mathit{Deg}_{\geq 2}} \rvert$ is at most~$2 e^{- \frac{c_1}{3 \eps^2}} + \delta / 2 \leq \delta$. Therefore, with probability~$1 - \delta$, the minimum between the two estimates is~$\frac{3 (n + \lvert \widehat{\textit{Deg}_1(F)} \rvert)}{8}$. As shown before,~$\lvert \widehat{\textit{Deg}_{1}(F)} \rvert$ is a~$1 \pm \eps$ estimate of~$\lvert \textit{Deg}_{1}(F) \rvert$, hence the return estimate is a~$4/3\cdot (1 \pm \eps)$ approximation of~$\beta(F)$ with probability~$1 - \delta$.

            The space usage of Algorithm \ref{alg:ind-num-main} is the maximum space usage of its three sub-algorithms, which is~$\widetilde{\mathcal{O}}(\sqrt{n})$.
        \end{proof}

        \subsubsection{One-pass and Two-pass Streaming Algorithms for Estimating the Domination number $\gamma$}

            Lemma~\ref{lemma:H2} and Theorem~\ref{theo:3approxDom-forest} allow us to $(3 \pm \eps)$-approximate~$\gamma$ in streaming forests using only polylog space.
            \begin{theorem}
            \label{theo:dom-3approx} 
                For every $\eps, \delta\in (0,1)$ and forest~$F$, the domination number $\gamma$ can be~$(3 \pm \eps)$-approximated in turnstile streams with probability~$(1 - \delta)$ and in $\mathcal{O}(\log^{\mathcal{O}(1)} n \cdot \log \delta^{-1})$ space.
            \end{theorem}

            In addition, reusing the above algorithms for non-leaf vertices and support vertices (Subroutine~\ref{alg:SLarge} and~\ref{alg:dom-H2-Small}), and following a main algorithm similar to the main algorithm (Algorithm~\ref{alg:ind-num-main}) for independence number ($\beta$), we are able to $2$-approximate domination number ($\gamma$) and $3/2$-approximate matching number ($\phi$) in forest streams with~$\widetilde{\mathcal{O}}(\sqrt{n})$ space and two passes. Our main algorithm (Algorithm \ref{alg:dom-num-main}) and theorem (Theorem \ref{theo:dom-2approx}) for approximating the domination number are shown below.

            \begin{algorithm}
                \caption{Main Algorithm for Estimating~$\gamma$}\label{euclid}
                \begin{algorithmic}[1]

                    \State{\textbf{Input}: error rate~$\eps$ and fail rate~$\delta$}

                    \State{\textbf{Initialization}: Set~$K_1 = \sqrt{n}$,~$K_2 = 12 \sqrt{n}$,~$c_1 = 3 \ln (6 \delta^{-1})$,~$c_2 = \delta^{-1}$,~$\eps_1 = \eps$}

                    \State{Run the following algorithms concurrently:}
                    \State{1. Algorithm for estimating~$\lvert \mathit{Deg}_{\geq 2} \rvert$ with~$\eps$ and~$\delta / 2$, denote the returned result as~$\lvert \widehat{\mathit{Deg}_{\geq 2}} \rvert$}
                    \State{2. Subroutine \ref{alg:SLarge} with~$K_1$,~$c_1$, and~$\eps_1$, denote the returned result as~$\lvert \widehat{S} \rvert$}
                    \State{3. Subroutine \ref{alg:dom-H2-Small} with~$K_2$ and~$c_2$, denote the returned result as~$\lvert \widehat{S} \rvert'$ and~$\lvert \widehat{\mathit{Deg}_{\geq 2}} \rvert'$}

                    \If{Subroutine \ref{alg:dom-H2-Small} returns \textit{FAIL}}
                        \State{\textbf{Return}~$\widehat{\gamma} = \max \{ \frac{2 \lvert \widehat{\mathit{Deg}_{\geq 2}} \rvert}{3}, \frac{\lvert \widehat{\mathit{Deg}_{\geq 2}} \rvert + \lvert \widehat{S} \rvert}{2} \}$}
                    \Else{}
                        \State{\textbf{Return}~$\widehat{\gamma} = \max \{ \frac{2 \lvert \widehat{\mathit{Deg}_{\geq 2}} \rvert'}{3}, \frac{\lvert \widehat{\mathit{Deg}_{\geq 2}} \rvert' + \lvert \widehat{S} \rvert'}{2} \}$}
                    \EndIf

                \end{algorithmic}
            \label{alg:dom-num-main}
            \end{algorithm}

            \begin{theorem}
            \label{theo:dom-2approx}
                For every $\eps, \delta\in (0,1)$ and forest~$F$ (with no isolated vertices), Algorithm~\ref{alg:dom-num-main} estimates the domination number~$\gamma(F)$ in a turnstile stream within factor~$2\cdot (1 \pm \eps)$ with probability~$(1 - \delta)$ in~$\widetilde{\mathcal{O}}(\sqrt{n})$ space and two passes.
            \end{theorem}
            \begin{proof}
                By Theorem \ref{theo:dom-num_2approx} and Corollary \ref{coro:2approxDS-nohusu}, the max between~$\frac{2 \lvert \mathit{Deg}_{\geq 2}(F) \rvert}{3}$ and~$\frac{\lvert \mathit{Deg}_{\geq 2}(F) \rvert + \lvert \supp(F) \rvert}{2}$ is guaranteed to be a $2$ approximation of the domination number~$\gamma(F)$. Thus it suffices to show that we could approximate both terms well, or we could approximate the larger term well.

                When~$\lvert \mathit{Deg}_{\geq 2}(F) \rvert \leq K_2 = 12 \sqrt{n}$, by Lemma \ref{lemma:dom-H2-Small}, Subroutine \ref{alg:dom-H2-Small} succeeds with probability~$1 - 1/c_2 = 1 - \delta$. Thus, by the design of the algorithm, we return our estimate at line 10. By Lemma \ref{lemma:dom-H2-Small}, conditional on successful returning, Subroutine \ref{alg:dom-H2-Small} gives exact values of~$\lvert \mathit{Deg}_{\geq 2}(F) \rvert$ and~$\lvert \supp(F) \rvert$. Hence we have a~$2$-approximation of~$\gamma(F)$.

                When~$\lvert \mathit{Deg}_{\geq 2}(F) \rvert > 12 \sqrt{n}$, there are two cases to consider. If~$\lvert \supp(F) \rvert \geq \sqrt{n}$, then the algorithm might return the estimate of~$\frac{2 \lvert \mathit{Deg}_{\geq 2}(F) \rvert}{3}$ or the estimate of~$\frac{\lvert \mathit{Deg}_{\geq 2}(F) \rvert + \lvert \supp(F) \rvert}{2}$. On the one hand, if the estimate of~$\frac{2 \lvert \mathit{Deg}_{\geq 2}(F) \rvert}{3}$ is returned, by Theorem \ref{lemma:H2}, we know that the algorithm on line 4 outputs a~$(1 \pm \eps)$ estimate of~$\lvert \mathit{Deg}_{\geq 2}(F) \rvert$ with probability~$1 - \delta / 2$. Hence we obtain a~$2 (1 \pm \eps)$ approximation of~$\gamma$.

                On the other hand, if the estimate of~$\frac{\lvert \mathit{Deg}_{\geq 2}(F) \rvert + \lvert \supp(S) \rvert}{2}$ is returned by our algorithm, we know that by Lemma \ref{lemma:SLarge}, Subroutine \ref{alg:SLarge} gives a ($1 \pm \eps$) estimate of~$\lvert \supp(S) \rvert$ with probability

                \begin{displaymath}
                    1 - 3 e^{-c_1 / 3} = 1 - 3 e^{- \ln (6 \delta^{-1})} = 1 - \frac{\delta}{2}\,.
                \end{displaymath}

                Applying a union bound over successfully returning an estimate of~$\lvert \supp(F) \rvert$ and an estimate of~$\lvert \mathit{Deg}_{\geq 2}(F) \rvert$, the probability of failure is at most~$\delta$. And the error rate is still~$1 \pm \eps$ as we are summing up the two terms. Hence, the returned estimate is still a~$2 (1 \pm \eps)$ approximation of~$\gamma$.

                Lastly, if~$\lvert \supp(F) \rvert < \sqrt{n}$, then by Lemma \ref{lemma:SLarge},~$\lvert \widehat{S} \rvert \leq 2 \sqrt{n}$ with probability~$1 - 2 e^{- \frac{c_1}{3}}$. By Theorem \ref{lemma:H2}, if~$\eps < 1/2$, then
                \begin{displaymath}
                    \lvert \widehat{\mathit{Deg}_{\geq 2}} \rvert \geq (1 - \eps) \lvert \mathit{Deg}_{\geq 2}(F) \rvert > \frac{\lvert \mathit{Deg}_{\geq 2}(F) \rvert}{2}
                \end{displaymath}
                with probability~$1 - \delta / 2$. Since we have~$\lvert \mathit{Deg}_{\geq 2}(F) \rvert > 12 \sqrt{n}$, we can bound the probability of~$3 \lvert \widehat{S} \rvert \leq \lvert \widehat{\mathit{Deg}_{\geq 2}} \rvert$ as

                \begin{displaymath}
                \begin{aligned}
                    \Pr \big[ 3 \lvert \widehat{S} \rvert \leq \lvert \widehat{\mathit{Deg}_{\geq 2}} \rvert \big]
                    &= \Pr \big[ \lvert \widehat{S} \rvert \leq 2 \sqrt{n} \wedge \lvert \widehat{\mathit{Deg}_{\geq 2}} \rvert > 6 \sqrt{n} \big] \\
                    &\geq 1 - \Pr \big[ \lvert \widehat{S} \rvert > 2 \sqrt{n} \big] - \Pr \big[ \lvert \widehat{\mathit{Deg}_{\geq 2}} \rvert \leq 6 \sqrt{n} \big] \\
                    &\geq 1 - 2 e^{- \frac{c_1}{3}} - \delta / 2 \\
                    &\geq 1 - \delta\,,
                \end{aligned}
                \end{displaymath}
                where we apply the union bound in the first inequality. Therefore, with probability at least~$1 - \delta$, the maximum between the two estimates,~$\frac{2 \lvert \widehat{\mathit{Deg}_{\geq 2}} \rvert}{3}$ and~$\frac{\lvert \mathit{Deg}_{\geq 2} \rvert + \lvert S \rvert}{2}$, is~$\frac{2 \lvert \widehat{\mathit{Deg}_{\geq 2}} \rvert}{3}$. As shown before,~$\lvert \widehat{\mathit{Deg}_{\geq 2}} \rvert$ is a~$1 \pm \eps$ estimate of~$\lvert \mathit{Deg}_{\geq 2}(F) \rvert$, hence the returned estimate is a~$2\cdot (1 \pm \eps)$ approximation of~$\gamma(F)$ with probability~$1 - \delta$.

                The space usage of Algorithm \ref{alg:dom-num-main} is the maximum space usage of its three sub-algorithms, which is~$\widetilde{\mathcal{O}}(\sqrt{n})$.
            \end{proof}

        \subsubsection{One-pass and Two-pass Streaming Algorithms for Estimating the Matching number $\phi$}

            Similarly, combining Lemma~\ref{lemma:H2} and Theorem~\ref{theo:2approx-matching-forest}, we are able to $(2 \pm \eps)$-approximate~$\phi$ in streaming forests using only polylog space.
            \begin{theorem}
            \label{theo:mat-2approx}
                For every~$\eps, \delta\in (0,1)$ and forest~$F$, the matching number~$\phi$ can be~$(2 \pm \eps)$-approximated in turnstile streams with probability~$(1 - \delta)$ and in $\mathcal{O}(\log^{\mathcal{O}(1)} n \cdot \log \delta^{-1})$ space.
            \end{theorem}

            The approximation ratio can be further improved to $3/2$ when we allow an additional passes and using a total of~$\widetilde{\mathcal{O}}(\sqrt{n})$ space. Our main algorithm (Algorithm \ref{alg:matching-num-main}) and theorem (Theorem \ref{theo:mat-3/2approx}) for approximating the matching number are shown below.

            \begin{algorithm}
                \caption{Main Algorithm for Estimating~$\phi$}
                \begin{algorithmic}[1]

                    \State{\textbf{Input}: error rate~$\eps$ and fail rate~$\delta$}

                    \State{\textbf{Initialization}: Set~$K_1 = \sqrt{n}$,~$K_2 = 8 \sqrt{n}$,~$c_1 = 3 \ln (6 \delta^{-1})$,~$c_2 = \delta^{-1}$,~$\eps_1 = \eps$}

                    \State{Run the following algorithms concurrently:}
                    \State{1. Algorithm for estimating~$\lvert \mathit{Deg}_{\geq 2} \rvert$ with~$\eps$ and~$\delta / 2$, denote the returned result as~$\lvert \widehat{\mathit{Deg}_{\geq 2}} \rvert$}
                    \State{2. Degree-counting algorithms to determine the number of tree components, $C$.}
                    \State{3. Subroutine \ref{alg:SLarge} with~$K_1$,~$c_1$, and~$\eps_1$, denote the returned result as~$\lvert \widehat{S} \rvert$}
                    \State{4. Subroutine \ref{alg:dom-H2-Small} with~$K_2$ and~$c_2$, denote the returned result as~$\lvert \widehat{S} \rvert'$ and~$\lvert \widehat{\mathit{Deg}_{\geq 2}} \rvert'$}

                    \If{Subroutine \ref{alg:dom-H2-Small} returns \textit{FAIL}}
                        \State{\textbf{Return}~$\widehat{\phi} = \max \{ \frac{3 (\lvert \widehat{\mathit{Deg}_{\geq 2}} \rvert + C)}{4}, \frac{\lvert \widehat{\mathit{Deg}_{\geq 2}} \rvert + \lvert \widehat{S} \rvert}{2} \}$}
                    \Else{}
                        \State{\textbf{Return}~$\widehat{\phi} = \max \{ \frac{3 (\lvert \widehat{\mathit{Deg}_{\geq 2}} \rvert' + C)}{4}, \frac{\lvert \widehat{\mathit{Deg}_{\geq 2}} \rvert' + \lvert \widehat{S} \rvert'}{2} \}$}
                    \EndIf

                \end{algorithmic}
            \label{alg:matching-num-main}
            \end{algorithm}

            \begin{theorem}
            \label{theo:mat-3/2approx}
                For every $\eps, \delta\in (0,1)$ and forest~$F$ (with no isolated vertices), Algorithm~\ref{alg:ind-num-main} estimates the matching number~$\phi(F)$ in a turnstile stream within factor~$3/2\cdot (1 \pm \eps)$ with probability~$(1 - \delta)$ in~$\widetilde{\mathcal{O}}(\sqrt{n})$ space and two passes.
            \end{theorem}

            \begin{proof}

                By Theorem \ref{matching_Forest_FullBound} and Corollary \ref{coro:4/3approxMM-noS}, the max between~$\frac{3 (\lvert \mathit{Deg}_{\geq 2} \rvert + C)}{4}$ and~$\frac{\lvert \mathit{Deg}_{\geq 2} \rvert + \lvert \supp \rvert}{2}$ is guaranteed to be a $3/2$ approximation of the matching number~$\phi$. Thus it suffices to show that we could approximate both terms well, or we could approximate the larger term well.

                When~$\lvert \mathit{Deg}_{\geq 2} \rvert \leq K_2 = 8 \sqrt{n}$, by Lemma \ref{lemma:dom-H2-Small}, Subroutine \ref{alg:dom-H2-Small} succeeds with probability~$1 - 1/c_2 = 1 - \delta$. Thus, by the design of the algorithm, we return our estimate at line 11. By Lemma \ref{lemma:dom-H2-Small}, conditional on successful returning, Subroutine \ref{alg:dom-H2-Small} gives exact values of~$\lvert \mathit{Deg}_{\geq 2} \rvert$ and~$\lvert \supp \rvert$. Hence we have a~$3/2$-approximation of~$\phi$.

                When~$\lvert \mathit{Deg}_{\geq 2} \rvert > 8 \sqrt{n}$, there are two cases to consider. If~$\lvert \supp \rvert \geq \sqrt{n}$, then the algorithm might return the estimate of~$\frac{3 (\lvert \mathit{Deg}_{\geq 2} \rvert + C)}{4}$ or the estimate of~$\frac{\lvert \mathit{Deg}_{\geq 2} \rvert + \lvert \supp \rvert}{2}$. Note $C$ is an exact estimate of the number of tree components. On the one hand, if the estimate of~$\frac{3 (\lvert \mathit{Deg}_{\geq 2} \rvert + C)}{3}$ is returned, by Theorem \ref{lemma:H2}, we know that the algorithm on line 4 outputs a~$(1 \pm \eps)$ estimate of~$\lvert \mathit{Deg}_{\geq 2} \rvert$ with probability~$1 - \delta / 2$. Hence we obtain a~$3/2 (1 \pm \eps)$ approximation of~$\phi$. On the other hand, if the estimate of~$\frac{\lvert \mathit{Deg}_{\geq 2} \rvert + \lvert \supp \rvert}{2}$ is returned by our algorithm, we know that by Lemma \ref{lemma:SLarge}, Subroutine \ref{alg:SLarge} gives a ($1 \pm \eps$) estimate of~$\lvert \supp \rvert$ with probability
                \begin{displaymath}
                    1 - 3 e^{-c_1 / 3} = 1 - 3 e^{- \ln (6 \delta^{-1})} = 1 - \frac{\delta}{2}\,.
                \end{displaymath}

                Apply a union bound over successfully returning an estimate of~$\lvert \supp \rvert$ and an estimate of~$\lvert \mathit{Deg}_{\geq 2} \rvert$, the probability of failure is at most~$\delta$. And the error rate is still~$1 \pm \eps$ as we are summing up the two terms. Hence, the returned estimate is still a~$3/2 (1 \pm \eps)$ approximation of~$\phi$.

                Lastly, if~$\lvert \supp \rvert < \sqrt{n}$, then by Lemma \ref{lemma:SLarge},~$\lvert \widehat{S} \rvert \leq 2 \sqrt{n}$ with probability~$1 - 2 e^{- \frac{c_1}{3}}$. By Theorem \ref{lemma:H2}, if~$\eps < 1/2$, then
                \begin{displaymath}
                    \lvert \widehat{\mathit{Deg}_{\geq 2}} \rvert \geq (1 - \eps) \lvert \mathit{Deg}_{\geq 2} \rvert > \frac{\lvert \mathit{Deg}_{\geq 2} \rvert}{2}
                \end{displaymath}
                with probability~$1 - \delta / 2$. Since we have~$\lvert \mathit{Deg}_{\geq 2} \rvert > 8 \sqrt{n}$, we can bound the probability of~$2 \lvert \widehat{S} \rvert \leq \lvert \widehat{\mathit{Deg}_{\geq 2}} \rvert$ as
                \begin{displaymath}
                \begin{aligned}
                    \Pr \big[ 2 \lvert \widehat{S} \rvert \leq \lvert \widehat{\mathit{Deg}_{\geq 2}} \rvert \big]
                    &= \Pr \big[ \lvert \widehat{S} \rvert \leq 2 \sqrt{n} \wedge \lvert \widehat{\mathit{Deg}_{\geq 2}} \rvert > 4 \sqrt{n} \big] \\
                    &\geq 1 - \Pr \big[ \lvert \widehat{S} \rvert > 2 \sqrt{n} \big] - \Pr \big[ \lvert \widehat{\mathit{Deg}_{\geq 2}} \rvert \leq 4 \sqrt{n} \big] \\
                    &\geq 1 - 2 e^{- \frac{c_1}{3}} - \delta / 2 \\
                    &\geq 1 - \delta\,,
                \end{aligned}
                \end{displaymath}
                where we apply the union bound in the first inequality. Therefore, with probability at least~$1 - \delta$, the maximum between the two estimates,~$\frac{3 (\lvert \widehat{\mathit{Deg}_{\geq 2}} \rvert + C)}{4}$ and~$\frac{\lvert \mathit{Deg}_{\geq 2} \rvert + \lvert S \rvert}{2}$, is~$\frac{3 (\lvert \widehat{\mathit{Deg}_{\geq 2}} \rvert + C)}{4}$. As shown before,~$\lvert \widehat{\mathit{Deg}_{\geq 2}} \rvert$ is a~$1 \pm \eps$ estimate of~$\lvert \mathit{Deg}_{\geq 2} \rvert$, hence the returned estimate is a~$3/2 \cdot (1 \pm \eps)$ approximation of~$\phi$ with probability~$1 - \delta$.

                The space usage of Algorithm \ref{theo:mat-3/2approx} is the maximum space usage of its four sub-algorithms, which is~$\widetilde{\mathcal{O}}(\sqrt{n})$.
            \end{proof}
\section{Conclusion \& Open Questions}

We designed \cwei to estimate the Caro-Wei bound.
It improves on  Halld{\'o}rsson et al.~\cite{halldorsson2016streaming}, to match Cormode et al.~\cite{cormode2018approximating}, also reporting a \emph{solution} in online streaming. A dedicated, write-only \emph{solution space} admits computing $\mathcal{O}(n)$-size solutions, a likely fruitful future direction.

We then invoked the notion of \emph{support vertices} from structural graph theory~\cite{delavina2010total} in the streaming model. We hence approximated graph parameters, via sketches to estimate the number of leaves, non-leaves and support vertices. This progresses towards meeting the approximation ratios for the lower bounds of Esfandiari et al.~\cite{esfandiari2018streaming}.
Support vertices could also help estimate  other streamed graph parameters in forests and other sparse graphs.

\section*{Acknowledgement}

We would like to thank Robert Krauthgamer for helpful discussions.

\section*{Author note}
Xiuge Chen is now with Optiver, Sydney. Patrick Eades is now with The University of Sydney.

\renewcommand{\baselinestretch}{1}

\newpage
\bibliographystyle{splncs04}
\bibliography{bibliography}

\renewcommand{\baselinestretch}{1}

\end{document}